\documentclass[runningheads]{llncs}

\usepackage[T1]{fontenc}

\usepackage{amsmath}
\usepackage{amssymb}
\usepackage{enumerate}
\usepackage{color}
\usepackage{cite}
\usepackage{booktabs}
\usepackage{xspace}
\usepackage{marvosym}
\usepackage{definitions}
\usepackage[hypertexnames=false]{hyperref}

\begin{document}

\title{The Computability Path Order for Beta-Eta-Normal Higher-Order Rewriting \\
  (Full Version)}

\titlerunning{CPO for Beta-Eta-Normal Higher-Order Rewriting}

\author{Johannes Niederhauser\textsuperscript{(\Letter)}%
\orcidID{0000-0002-8662-6834} \and
Aart Middeldorp\orcidID{0000-0001-7366-8464}}

\authorrunning{J.~Niederhauser and A.~Middeldorp}

\institute{Department of Computer Science,
University of Innsbruck, Innsbruck, Austria
\email{\{johannes.niederhauser,aart.middeldorp\}@uibk.ac.at}}

\maketitle

\begin{abstract}
We lift the computability path order and its extensions from
plain higher-order rewriting to higher-order rewriting on
$\beta\eta$-normal forms where matching modulo $\beta\eta$ is
employed. The resulting order NCPO is shown to be useful on
practical examples. In particular, it can handle systems where its
cousin NHORPO fails even when it is used together with the
powerful transformation technique of neutralization.
We also argue
that automating NCPO efficiently is straightforward using SAT/SMT
solvers whereas this cannot be said about the transformation
technique of neutralization. Our prototype implementation
supports automatic termination proof search for NCPO and is also
the first one to automate NHORPO with neutralization.
\end{abstract}

\section{Introduction}

Higher-order rewriting is known for its abundance of different
formalisms where it is often unclear whether results for a particular
notion of higher-order rewriting can be transferred to other kinds of
higher-order rewriting \cite{K12}. In this paper, we are only
concerned with one particular formalism, namely a slightly
modified version of the higher-order rewrite systems (HRSs) {\`a}
la Nipkow \cite{MN98}. However, as the goal of this paper is to lift a
particular termination method from one formalism to another, a short
discussion about differences and similarities of the used formalisms
is necessary. First and foremost, when we talk about higher-order
rewriting in this paper, we mean the particular flavor of higher-order
rewriting which considers terms of the simply-typed lambda calculus as
objects to be rewritten. Within this class of higher-order rewrite
formalisms, we further distinguish between \emph{plain} and
\emph{normal} higher-order rewriting: Plain higher-order rewriting
uses plain syntactic matching and views lambda calculus'
$\beta$-reduction as a proper rewrite rule which has to be oriented
for termination. On the other hand, normal higher-order rewriting uses
matching modulo $\beta\eta$, which means that terms which are
$\beta\eta$-equivalent should also be equivalent with respect to the
reduction order that is used for establishing termination. Nipkow's
HRSs are a premier example of normal higher-order rewriting. HRSs are
an interesting formalism for automated reasoning as they can directly
represent the terms of higher-order logic and as such are part of the
meta-theory of interactive theorem provers like Isabelle
\cite{NPW02}. Furthermore, they come with a critical pair lemma
\cite{MN98} which is vital for automated equational reasoning with
techniques such as completion \cite{HMSW19}.
In particular, automated confluence analysis of higher-order rewriting
typically uses HRSs, as witnessed in the Confluence Competition
(2015--2020).\footnote{https://project-coco.uibk.ac.at/}

While many termination methods are available for plain higher-order
rewriting \cite{JR07,K12,BJR15}, the situation for normal
higher-order rewriting is different. For a long time, powerful
termination methods for HRSs were only based on van de Pol's results
\cite{vdP96} which employ the monotone algebra approach \cite{Z94},
i.e., semantic termination arguments which are difficult to automate
except for more specialized cases like polynomial interpretations of
fixed shape \cite{FK12}. An early version of the
\emph{higher-order recursive path order} (HORPO) which is
designed for plain higher-order rewriting was adapted for the
important subclass of pattern HRSs in \cite{vR01}, an implementation
of this approach is available in the tool \csiho~\cite{N17}. Furthermore,
the formalism of \emph{algebraic functional systems with meta-variables}
(AFSMs) captures both plain and normal rewriting and therefore makes a
larger number of termination methods available for HRSs~\cite{K12}. An
implementation of these methods is available in the tool WANDA~\cite{K20}.
However, WANDA has not been optimized for HRSs and the theory developed
around it only establishes termination for systems where all left-hand
sides are patterns \cite{M91}. A powerful recursive path order
specifically designed for HRSs and without the restriction to patterns
as left-hand sides of rules was missing until Jouannaud and Rubio
lifted their definite version of HORPO \cite{JR07} from plain to
normal higher-order rewriting, resulting in NHORPO
\cite{JR15}. However, the important extension of HORPO by the
computability closure \cite{JR07} is not considered. Instead, a novel
transformation technique named \emph{neutralization} is
introduced. Its goal is to simplify a given system with respect to the
applicability of NHORPO without affecting the system's termination
behavior.

The aim of this paper is to develop a syntax-directed reduction order
for normal higher-order rewriting which is sufficiently powerful
compared to NHORPO with neutralization but easier to automate using
SAT/SMT solvers. Our starting point is the
\emph{computability path order} (CPO) \cite{BJR15} which is defined
for plain higher-order rewriting. CPO is an extension of HORPO which
incorporates computability closure \cite{BJO99} in a sophisticated
way. We follow the general approach from \cite{JR15}
to lift the extension of core CPO with accessible subterms and small
symbols to HRSs, resulting in the \emph{$\beta\eta$-normal
computability path order} (NCPO). In particular, we will show
that there are HRSs which NCPO can prove terminating but where
NHORPO (with or without neutralization) as well as the HORPO
implementations in WANDA~\cite{K20} and \csiho~\cite{N17} fail.

In designing a reduction order for normal higher-order rewriting, one
has to deal with the well-known problems that well-founded orders
which
are compatible with $\beta\eta$ cannot be \emph{monotone} (closed under
contexts) and \emph{stable} (closed under substitutions) in general as
the following examples show:

\begin{example}
\label{exa:challenges}
Consider $\m{a} > \m{b}$. If $>$ is closed under contexts, we should
also have $(\lambda y.\m{c})\m{a} > (\lambda y.\m{c})\m{b}$. If we
also want $>$ to be compatible with $\beta\eta$, we obtain
$\m{c} > \m{c}$ which means that $>$ cannot be well-founded. Now
consider $F x > x$. If $>$ is closed under substitutions, we should
also have $(\lambda x.x) \m{a} > \m{a}$. If we also want $>$ to be
compatible with $\beta\eta$, we obtain $\m{a} > \m{a}$ which again
means that $>$ cannot be well-founded.
\end{example}

In \cite{vdP96}, van de Pol models substitutions via contexts using
$\beta$-reduction, thereby eliminating the second problem. For the
first problem, two different orders $>_1$ and $>_2$ are introduced
and it is shown that $s >_1 t$ implies $C[s] >_2 C[t]$ for all terms
$s$, $t$ and contexts $C$. Since the relations between $>_1$ and $>_2$
are intrinsically semantic, it is not clear how this approach can be
transferred to path orders. In \cite{JR15}, both problems were
solved by restricting the order $>$ to $\beta\eta$-normal forms and
layering it with its plain version ($\sqsupset$) in a way such that
$s > t$ implies $s\sigma{\D} \sqsupset^+ t\sigma{\D}$ for all
substitutions $\sigma$ where $s{\D}$ denotes the $\beta\eta$-normal
form of a term $s$. Then, monotonicity and well-foundedness of the
plain version can be used in order to show that $(>,\sqsupset)$
is a reduction order. We will utilize this approach in order to
define a reduction order $(>,\sqsupset)$ where $>$ is NCPO and
$\sqsupset$ is CPO.

The remainder of this paper is structured as follows.
\secref{preliminaries} introduces the basic concepts including the
type of higher-order rewriting which we are concerned with in this
paper as well as an appropriate notion of $\beta\eta$-normal
higher-order
reduction orders. \secref{ingredients} describes the necessary
ingredients of our order
regarding types, the important concept of nonversatile terms as well as
accessible subterms. Based on this, NCPO is introduced in \secref{order}
and proven to be a correct termination method for HRSs in
\secref{proof}. Finally, our prototype implementation of NCPO is
described and compared to NHORPO with neutralization in
\secref{nhorpo} before we conclude in \secref{conclusion}.

\section{Preliminaries}
\label{sec:preliminaries}

In this paper, we consider higher-order rewriting on simply-typed
lambda terms \cite{C40,B92}. Given a set $\xB$ of base types, the set
of simple types $\xT$ is defined inductively: If $a \in \xB$ then
$a \in \xT$, and for $U, V \in \xT$ also $U \to V \in \xT$. We follow
the usual convention that the function space constructor $\to$ is
right-associative, i.e., $a \to b \to c$ denotes $a \to (b \to c)$.
Throughout this text, lowercase letters $a, b, c, \dots$ denote base
types while upper case letters $T, U, V, \dots$ denote arbitrary types.
For each type $U \in \xT$ we consider an infinite set of variables
$\xV_U$ as well as a set of function symbols $\xF_U$ where
$\xV_U \cap \xF_U = \varnothing$ and
$\xV_U \cap \xV_V = \xF_U \cap \xF_V = \varnothing$ for $V \neq U$.
We denote the set of all variables by
$\xV = \bigcup \SET{\xV_U \mid U \in \xT}$. The set of all function
symbols $\xF = \bigcup \SET{\xF_U \mid U \in \xT}$ is referred to as the
\emph{signature}. We associate with each function symbol $f$ an
arity $\ar{f} \in \mathbb{N}$. The set of well-typed lambda terms
of a type $U$ ($\Lambda_U$) is defined as follows:
\begin{itemize}
\item
$\xV_U \subseteq \Lambda_U$,
\item
if $f \in \xF_{T_1 \to \cdots \to T_{\ar{f}} \to U}$ and
$t_i \in \Lambda_{T_i}$ for $1 \leq i \leq \ar{f}$ then
$f(\seq{t}) \in \Lambda_U$,
\item
if $x \in \xV_U$ and $s \in \Lambda_V$ then
$\lambda x.s \in \Lambda_{U \to V}$,
\item
if $s \in \Lambda_{U \to V}$ and $t \in \Lambda_U$ then
$st \in \Lambda_V$.
\end{itemize}
The set $\Lambda = \bigcup \SET{\Lambda_U \mid U \in \xT}$ contains all
well-typed lambda terms. We define the function
$\tau\colon \Lambda \to \xT$ as $\tau(s) = U$ if $s \in \Lambda_U$.
Throughout this paper we only consider well-typed lambda terms.
We also follow the convention that
application is left-associative, i.e. $stu$ denotes $(st)u$. Note
that our definition of lambda terms with function symbols of fixed
arity neither poses a limitation nor adds expressive power as we can
always set $\m{ar}(f) = 0$ for all $f \in \xF$ (the ``return
type'' of a function symbol does not have to be a base type) and
denote $f(\seq{t})$ by the corresponding application
$f t_1 \cdots t_n$. The addition of function symbols with fixed
arity is solely needed for an adequate definition of a recursive path
order on lambda terms. We sometimes use the shorthand
$f(\bar{t})$ to denote the application of $f$ to the list of
arguments $\bar{t}$.

We write $\fv{s}$ for the set of free variables of a term $s$. The
term $s[x/t]$ denotes the term where all free occurrences of $x$ have
been replaced by $t$ without capturing the free variables of $t$
(\emph{capture-avoiding substitution}). Due to the fact that
infinitely many variables are available for each type, this can always
be done by renaming the bound variables accordingly
(\emph{$\alpha$-renaming}). In the remainder, we abstract away from
this technicality by viewing lambda terms as equivalence classes
modulo $\alpha$-renaming.

Two fundamental concepts in lambda calculus are the notions of
$\beta$- and $\eta$-reduction which are defined as the rule schemas
$(\lambda x.s)t \Rb{\beta} s[x/t]$ and $\lambda x.ux \Rb{\eta} u$ if
$x \notin \fv{u}$. Note that both $\beta$- and $\eta$-reduction
preserve types. Every term $s$ has a unique $\beta\eta$-normal form
which we denote by $s{\D}$. Rewriting to $\beta\eta$-normal
form and the set of $\beta\eta$-normal forms are denoted by $\Rnbe$
and $\nfbe \subseteq \Lambda$, respectively.

A substitution $\sigma$ is a mapping from variables to terms of the
same type where $\dom{\sigma} = \SET{x \mid \sigma(x) \neq x}$ is finite.
We often write $\sigma$ as a set of variable bindings. Given a
substitution $\sigma = \SET{x_1 \mapsto t_1,\dots,x_n \mapsto t_n}$, we
define $s\sigma$ as the simultaneous capture-avoiding substitution
$s[x_1/t_1,\dots,x_n/t_n]$. The free variables of a substitution are
defined as follows:
$\fv{\sigma} = \bigcup \SET{\fv{\sigma(x)} \mid x \in \dom{\sigma}}$. A
substitution is said to be \emph{away from} a finite set of variables
$X$ if $(\dom{\sigma} \cup \fv{\sigma}) \cap X = \varnothing$. We
follow the convention that the postfix operations of substitution
application as well as $\D$ bind stronger than lambda
abstractions and applications, i.e.,
$u\sigma{\D}v\sigma{\D} =
(u\sigma{\D})(v\sigma{\D})$ and
$\lambda x.t\sigma{\D} = \lambda x.(t\sigma{\D})$.
If $\bar{t} = (\seq{t})$, we allow ourselves to write $\bar{t}\sigma$
($\bar{t}\sigma{\D}$) as a shorthand for
$(t_1\sigma,\dots,t_n\sigma)$
($(t_1\sigma{\D},\dots,t_n\sigma{\D})$).

Contexts $C$ are lambda terms which contain exactly one occurrence of
the special symbol $\square$ which can assume any type. We write
$C[s]$ for the lambda term $C$ where $\square$ is replaced by $s$
without employing capture-avoiding substitution.
A binary relation $R \subseteq \Lambda \times \Lambda$ is
\emph{monotone} if $s \mathrel{R} t$ implies $C[s] \mathrel{R} C[t]$
for every context $C$. Furthermore, we say that a term
$t$ is a \emph{subterm} of $s$ ($s \subteq t$) if there exists a
context $C$ such that $s = C[t]$ and define the proper subterm
relation $s \subt t$ as $s \subt t$ and $s \neq t$. Since we view
lambda terms as equivalence classes modulo $\alpha$-renaming, the
subterm relation is also defined modulo $\alpha$-renaming. Hence, we
have e.g.~$\lambda x.t \subt t[x/z]$ for a fresh variable $z$.
Now, we are able to define the notion of
higher-order rewriting we are concerned with in this paper.

\begin{definition}
A \emph{rewrite rule} is a tuple $\ell \R r$ with $\ell, r \in \nfbe$
where $\ell$ is not of the form $xs_1 \cdots s_n$,
$\tau(\ell) = \tau(r)$ and $\fv{r} \subseteq \fv{\ell}$. A
\emph{higher-order rewrite system} \textup{(HRS)} is a set of
rewrite rules. Given an \textup{HRS} $\xR$, there is a \emph{rewrite step}
$s \RbR t$ if there exist a rule $\ell \R r \in \xR$, a
substitution $\sigma$ and a context $C$ such that
$s = C[\ell\sigma{\D}] \in \nfbe$ and
$t = C[r\sigma{\D}]{\D}$.
\end{definition}

Note that ${\RbR} \subseteq \nfbe \times \nfbe$ as
$C[\ell\sigma{\D}], C[r\sigma{\D}]{\D} \in \nfbe$ by
definition. Hence, both rules and rewrite steps only consider terms in
their unique $\beta\eta$-normal form whereas matching is performed
modulo $\beta\eta$. In the original definition of HRSs \cite{MN98},
the long $\beta\eta$-normal form based on $\eta$-expansion instead of
$\eta$-reduction is used. Using $\beta\eta$-normal
forms instead of long $\beta\eta$-normal forms leads to a different
rewrite relation and might change the termination behavior of a
given rewrite system (cf.~\exaref{betaEtaLong}).
The reason of our deviation from the standard
definition of Nipkow's HRSs is that higher-order recursive path
orders usually have a hard time dealing with lambdas, and we
wanted keep \thmref{termination} simple by using the same canonical
form for rewriting as well as higher-order reduction orders.
Furthermore, higher-order rewriting on $\beta\eta$-normal forms is
interesting in its own right \cite{J05,JL12}. Note that NHORPO
\cite{JR15} can accommodate any orientation of $\eta$ by adding an
appropriate wrapper around the core order which is only defined
on $\beta\eta$-normal forms. We expect that a similar strategy
works for NCPO, thus making it applicable to the original
formulation of HRSs by Nipkow. Another difference from the usual
definition of HRSs is that we allow rules of non-base type. However,
this is not new \cite{JL12}. In all other aspects, our definition is
equivalent to the original one given in \cite{MN98}.

\begin{example}
\label{exa:betaEtaLong}
Let $\xB = \SET{a}$ and consider the function symbols
$\m{c} \in \xF_a$, $\m{f} \in \xF_{a \to a}$ and
$\m{g} \in \xF_{a \to a \to a}$ as well as the variable
$x \in \xV_a$. The $\beta\eta$-normal term $\m{f}$ cannot be
rewritten using the HRS $\xR$ consisting of the single rule
$\m{f} x \R \m{g} x \m{c}$. However, for its $\eta$-long normal
form $\lambda y. \m{f} y$ we have
$\lambda y. \m{f} y \Rb{\xR} \lambda y. \m{g} y \m{c}$ in the
standard definition of HRSs.
\end{example}

Next, we recall some important definitions about relations which are
used throughout this paper. Given a binary relation $R$, $R^+$ and
$R^*$ denote its transitive and transitive-reflexive closure,
respectively.
The composition of a two binary relations $R$ and $S$ is
defined as follows: $a \mathrel{R} \cdot \mathrel{S} b$ if there
exists an element $c$ such that $a \mathrel{R} c$ and
$c \mathrel{S} b$.
A \emph{preorder} is a reflexive and
transitive relation. Given a relation $>$ we denote its inverse by
$<$ and its reflexive closure by $\geq$. Note that the reflexive
closure of a binary relation $R$ defined as $R \subseteq A \times A$
contains all elements of the set $A$ even when the strict part of $R$
assumes additional properties of the elements in $A$ which belong to
the relation.
A binary relation $R$ on a set $A$ is
\emph{well-founded} if there is no infinite sequence
$a_1 \mathrel{R} a_2 \mathrel{R} \cdots$ where $a_i \in A$ for
$i \in \mathbb{N}$. We say that an HRS $\xR$ is \emph{terminating} if
$\RbR$ is well-founded. We now define the notion of
$\beta\eta$-normal higher-order reduction orders.

\begin{definition}
\label{def:reductionOrder}
A \emph{$\beta\eta$-normal higher-order reduction order} is a pair
$(>,\sqsupset)$ which satisfies the following conditions:
\begin{itemize}
\item
${\sqsupset} \subseteq \Lambda \times \Lambda$ is a well-founded relation,
\item
${\sqsupset}$ is monotone,
\item
${\Rb{\beta}},{\Rb{\eta}} \subseteq {\sqsupset}$,
\item
$s > t$ implies $s\sigma{\D} \sqsupset^+ t\sigma{\D}$ for all
$s, t \in \nfbe$ and substitutions $\sigma$.
\end{itemize}
We often refer to the last condition as
\emph{$\beta\eta$-normal stability}.
An \textup{HRS} $\xR$ is \emph{compatible} with a $\beta\eta$-normal
higher-order reduction order $(>,\sqsupset)$ if $\ell >^+ r$
for all $\ell \R r \in \xR$.
\end{definition}

As in \cite{JR15}, the intuition behind this definition is that
$>$ will be used to orient the rules of HRSs while relying on the
termination proof of its plain variant $\sqsupset$. Despite calling
$(>,\sqsupset)$ an order, we do not demand transitivity of any of its
components. In the context of higher-order rewriting, this is
standard as $\sqsupset$ contains $\beta$-reduction which is
not transitive. By taking the identity substitution, we can see
that $\beta\eta$-normal stability implies ${>} \subseteq {\sqsupset^+}$.
The following theorem shows that $\beta\eta$-normal higher-order
reduction orders can be used to show termination of our flavor of HRSs.

\begin{theorem}
\label{thm:termination}
If an \textup{HRS} $\xR$ is compatible with a
$\beta\eta$-normal higher-order reduction order $(>,\sqsupset)$, then
$\xR$ is terminating.
\end{theorem}

\begin{proof}
For a proof by contradiction, consider an infinite rewrite sequence
$s_1 \RbR s_2 \RbR \cdots$. By definition,
$s_1 = C[\ell\sigma{\D}]$ and
$s_2 = C[r\sigma{\D}]{\D}$ for some
$\ell \R r \in \xR$ where $\sigma$ is a substitution and $C$ is a
context such that $C[\ell\sigma{\D}] \in \nfbe$. By
assumption, $\ell >^+ r$, so
$\ell\sigma{\D} \sqsupset^+ r\sigma{\D}$ follows from
$\beta\eta$-normal stability and we obtain
$C[\ell\sigma{\D}] \sqsupset^+ C[r\sigma{\D}]$ by
monotonicity of $\sqsupset$. Finally, $C[(\ell\sigma){\D}] \sqsupset^+
C[(r\sigma){\D}]{\D}$ since $\sqsupset$ contains
$\beta\eta$-reduction. Hence, we can transform the infinite
sequence $s_1 \RbR s_2 \cdots$ into the infinite descending sequence
$s_1 \sqsupset^+ s_2 \cdots$ which contradicts well-foundedness of
$\sqsupset$. Thus, $\xR$ is terminating. \qed
\end{proof}

\section{Ingredients of the Order}
\label{sec:ingredients}

We start by recalling the notion of nonversatile terms from
\cite[Definition 4.1]{JR15}. Intuitively, a term is nonversatile if the
application of any substitution together with the subsequent
$\beta\eta$-normalization only affects its proper subterms. As this
is needed in the inductive proof of $\beta\eta$-normal stability, an
NCPO comparison $s > t$ is only defined for nonversatile $s$. The
sufficient condition for nonversatility given in the subsequent lemma
is an improvement over the one given in \cite[Lemma 4.2]{JR15}.%
\footnote{In fact, subgoal 3 of Example 4.8 in \cite{JR15} cannot be
handled with the corresponding sufficient condition.}

\begin{definition}
The set $\Lambda_\nv \subseteq \nfbe$ of \emph{nonversatile} terms is
defined as follows:
\begin{itemize}
\item
$\xV \cap \Lambda_\nv = \varnothing$,
\item
$f(\bar{t}) \in \Lambda_\nv$ for all $f \in \xF$ and
$\bar{t} \in \Lambda^{\ar{f}}$,
\item
$uv \in \Lambda_\nv$ if
$(uv)\sigma{\D} = u\sigma{\D}v\sigma{\D}$
for all substitutions $\sigma$,
\item
$\lambda x.u \in \Lambda_\nv$ if
$(\lambda x.u)\sigma{\D} = \lambda x.u\sigma{\D}$
for all substitutions $\sigma$.
\end{itemize}
Terms which are not nonversatile are said to be \emph{versatile}.
\end{definition}

\begin{example}
The terms $\m{c} x$ and $\lambda x. \m{f}(y x)$ are nonversatile
while the terms $y x$ and $\lambda x. \m{f}(y x) x$ are versatile
as can be seen by taking the substitution
$\SET{y \mapsto \lambda z. \m{d}}$.
\end{example}

\begin{lemma}
\label{lem:nonversatile}
The following terms are nonversatile if they are $\beta\eta$-normal forms:
\begin{enumerate}[(i)]
\item
applied function symbols $f(\bar{t})$,
\item
applications $uv$ with nonversatile $u$,
\item
abstractions $\lambda x.ux$ where $ux \subteq vw$ implies that
$vw \in \Lambda_\nv$,
\item
abstractions $\lambda x.u$ with $u \neq vx$ where
\begin{itemize}
\item
$u \in \xV \cup \Lambda_\nv$,
\item
if $u = v(w_1w_2)$ then $w_1w_2 \in \Lambda_\nv$,
\item
if $u = v(\lambda y.w)$ then $\lambda y.w \in \Lambda_\nv$.
\end{itemize}
\end{enumerate}
\end{lemma}

Besides the usual inference rules on terms, reduction orders
derived from HORPO \cite{JR07} require appropriate orders on
types. The next definition recalls the notion of
admissible type orders from CPO \cite{BJR15}.

\begin{definition}
We define the left (right) argument relation on types $\subt_l$
\textup{($\subt_r$)} as follows: $T \to U \subt_l T$
\textup{($T \to U \subt_r U$)}. An
order $\succ_\xT$ on types is admissible if
\begin{itemize}
\item
${\subt_r} \subseteq {\succ_\xT}$,
\item
${\gtrdot} = ({\succ_\xT} \cup {\subt_l})^+$ is well-founded,
\item
if $T \to U \succ_\xT V$ then $U \succeq_\xT V$ or $V = T \to U'$ with
$U \succ_\xT U'$.
\end{itemize}
Given a type $T$ and $a \in \xB$ we write $a \gtrdot_\xB T$
\textup{($a \geqdot_\xB T$)} if $a \gtrdot b$
\textup{($a \geqdot b$)} for every
$b \in \xB$ occurring in $T$.
\end{definition}

\begin{lemma}[Lemma 2.3 from \cite{BJR15}]
\label{lem:admissibleOrdering}
Given a well-founded order $\succ_\xB$ on base types, let
$\succ_\xT$ be the smallest order on types containing $\succ_\xB$
and $\subt_r$ such that $V \succ_\xT V'$ implies
$U \to V \succ_\xT U \to V'$ for all $U,V,V'$. Then, $>_\xT$ is
admissible.
\end{lemma}

\begin{example}
Let $\xB = \{a,b,c\}$ where $a \succ_\xB b \succ_\xB c$ and consider
the order $\succ_\xT$ defined in the previous lemma. We have
$a \to b \succ_\xT a \to c$ because $b \succ_\xT c$. However,
$c \succ_\xT b \to b$ does not hold because $b \to b \subt_r b$ and
therefore $c \succ_\xT b \to b \succ_\xT b$ which contradicts
well-foundedness of $\succ_\xT$ together with $b \succ_\xT c$.
\end{example}

Unlike their first-order versions, higher-order recursive path
orders do not enjoy the subterm property because we cannot have
$\m{f}(\m{g}(x)) > x$ in general: If
$\m{f} \in \xF_{a \to a \to a}$ and
$\m{g} \in \xF_{(a \to a) \to a}$ for some $a \in \xB$ then this is
an encoding of the untyped lambda calculus (where $\m{f}$ represents
application and $\m{g}$ represents abstraction) and therefore not
terminating. Thus, if we want to recursively define
$f(\bar{t}) > v$ by $t_i \geq v$ for some $i$, we usually have to check
whether $\tau(t_i) \succeq_\xT \tau(v)$ holds for the given
admissible type order $\succ_\xT$. In particular, this means
that establishing $s > t$ by choosing some $s \subt s'$ and showing
$s' \geq t$ is only possible if we check that there is a weak
decrease in the admissible type order for all intermediate terms
in the recursive definition of $\subt$. CPO extends HORPO by allowing
these checks to be dismissed for the special case of accessible
subterms which are determined by type-related properties. To this
end, we start by introducing the concepts of (positive and negative)
base type positions in a type taken from \cite[Definition 7.2]{BJR15}.

\begin{definition}
The sets $\posp{T}$, $\posn{T}$ and $\posof{a}{T}$ of
\emph{positive base type positions},
\emph{negative base type positions} and \emph{positions of $a \in \xB$} in
a type $T \in \xT$ are inductively defined as sets of finite strings over
$\SET{1,2}$ (where $\epsilon$ denotes the empty string) as follows:
\begin{align*}
\posp{a} &= \posof{a}{a} = \SET{\epsilon} \qquad
\posn{a} = \varnothing \qquad
\posof{a}{b} = \varnothing ~ \text{if $a \neq b$} \\
\posof{a}{U \to V} &= \SET{1p \mid p \in \posof{a}{U}} \cup
\SET{2p \mid p \in \posof{a}{V}} \\
\posp{U \to V} &= \SET{1p \mid p \in \posn{U}} \cup
\SET{2p \mid p \in \posp{V}} \\
\posn{U \to V} &= \SET{1p \mid p \in \posp{U}} \cup
\SET{2p \mid p \in \posn{V}}
\end{align*}
\end{definition}

\begin{example}
Let $\xB = \SET{a,b}$ and consider $T = (a \to b) \to (a \to b)$. We
have $\posp{T} = \SET{11,22}$ and
$\posn{T} = \SET{12,21}$. Furthermore, $\posof{a}{T} = \SET{11,21}$
and $\posof{b}{T} = \SET{12,22}$.
\end{example}

Next, we define the notions of accessible arguments of function
symbols and basic base types as given in
\cite[Definitions 7.3 \& 7.4]{BJR15}.

\begin{definition}
With every $f \in \xF_{T_1 \to \cdots \to T_n \to a}$ we
associate a set $\acc{f}$ of \emph{accessible} arguments of
$f$ such that $i \in \acc{f}$ implies $a \geqdot_\xB T_i$
and $\posof{a}{T_i} \subseteq \posp{T_i}$ for all $1 \leq i \leq n$.
Furthermore, we say that $a \in \xB$ is \emph{basic} if the following
conditions hold:
\begin{itemize}
\item
$T \prec_\xT a$ implies that $T$ is a basic base type,
\item
for all $f \in \xF_{T_1 \to \cdots \to T_n \to a}$ and
$i \in \acc{f}$, $T_i = a$ or $T_i$ is a basic base type.
\end{itemize}
\end{definition}

Note that the condition $T \prec_\xT a$ is straightforward to check with
the admissible type order from \lemref{admissibleOrdering} as only
base types can be smaller than base types. In general, it is possible to
have base types which are bigger than function types while retaining
admissibility of the type order \cite{BJR15}.

Next, we define the notion of nonversatilely accessible subterms.
The definition closely follows the one given in
\cite[Definition 7.5]{BJR15} but with appropriate restrictions regarding
nonversatility.

\begin{definition}
We write $s \bsubt t$ if $t$ is a subterm of $s \in \Lambda_\nv$
reachable by nonversatile intermediate terms in the recursive definition
of $\subt$, $t$ is of basic base type and $\fv{t} \subseteq \fv{s}$, i.e.,
no bound variables have become free. Furthermore, $s \asubt t$ if there
are $a \in \xB$, $f \in \xF_{T_1 \to \cdots \to T_n \to a}$,
$s_i \in \Lambda_{T_i}$ for $1 \leq i \leq n$ and $j \in \acc{f}$
such that
$s = f(s_1,\dots,s_{\ar{f}}) s_{\ar{f}+1} \cdots s_n$ and
$s_j \asubteq t$. Here,
${\bsubt}, {\asubt} \subseteq \nfbe \times \nfbe$. A term $t$
is \emph{nonversatilely accessible} in a nonversatile term $s$ if
$s \bsubt t$ or $s \asubt t$.
\end{definition}

Finally, we define
the notion of structurally smaller terms with respect to a set
of variables $X$. It is an important ingredient of CPO with
accessible subterms as it facilitates a way of using the set $X$ with
which the order is parameterized in places where it would otherwise
lead to non-termination. Once again, we can immediately use the
corresponding definition given in \cite[Definition 7.8]{BJR15}.

\begin{definition}
Let $X$ be a finite set of variables. We say that a term $t$ is
\emph{structurally smaller} than a term $s$, written
$s \structsm{X} t$, if there are $a \in \xB$,
$\seq[k]{x} \in X$ and $u \in \Lambda$ such that
$\tau(s) = \tau(t) = a$, $t = u x_1 \cdots x_k$, $s \asubt u$ and
$\posof{a}{\tau(x_i)} = \varnothing$ for all $1 \leq i \leq k$.
Here, $\structsm{X} \subseteq \nfbe \times \Lambda$.
\end{definition}

Note that if $s \structsm{X} t$ then $t$ may not be in
$\beta\eta$-normal form. The following important result is
required for $\beta\eta$-normal stability and explains why nonversatility
is needed in $\bsubt$ but not in $\asubt$ and $\structsm{X}$ where it is
guaranteed by the original definition.

\begin{lemma}
\label{lem:subtNormStability}
The following statements hold:
\begin{enumerate}[(i)]
\item
If $s \bsubt t$ then $s\sigma{\D} \bsubt t\sigma{\D}$
for all substitutions $\sigma$.
\item
If $s \asubt t$ then $s\sigma{\D} \asubt t\sigma{\D}$
for all substitutions $\sigma$.
\item
If $s \structsm{X} t$ then $s\sigma{\D} \structsm{X} t'$ for some
$t' \Rnbe t\sigma{\D}$ whenever $\sigma$ is away from $X$.
\item
If $s \structsm{X} t$ and $t \in \Lambda_\nv$ then
$s\sigma{\D} \structsm{X} t\sigma{\D}$ for all
substitutions $\sigma$ away from $X$.
\end{enumerate}
\end{lemma}

\section{The Beta-Eta-Normal Computability Path Order}
\label{sec:order}

First, we briefly recapitulate the definition of the multiset and
lexicographic extension of orders. Given an order $>$, let
$(\seq{s}) >_{\lex} (\seq[m]{t})$ if there exists an
$i \leq \minrel{n}{m}$ such that $s_i > t_i$ and $s_j = t_j$ for all
$j < i$. Given two multisets $M$ and $N$ we write $M \gtrmul N$ if
$N = (M \setminus X) \uplus Y$ where $\varnothing \neq X \subseteq M$
and for all $y \in Y$ there exists an $x \in X$ such that $x > y$. It
is well-known that these extensions preserve well-foundedness.

In the following, let $\succ_\xT$ be an admissible order on types
and $\succsim_\xF$ a preorder on $\xF$ called \emph{precedence}
with a well-founded strict part
${\succ_\xF} = {\succsim_\xF} \setminus {\precsim_\xF}$ and the
equivalence relation
${\simeq_\xF} = {\succsim_\xF} \cap {\precsim_\xF}$. Furthermore, for
every $f \in \xF$ we fix a status $\stat{f} \in \SET{\mul,\lex}$
such that symbols equivalent in $\simeq_\xF$ have the same
status. We also assume that $\xF$ is partitioned into sets $\xFb$
and $\xFs$ of \emph{big} and \emph{small} symbols such that the
following conditions hold: $\xFb \cap \xFs = \varnothing$,
if $f \succsim_\xF g$ and $g \in \xFb$ then $f \in \xFb$, and
whenever $f \in \xF_{T_1 \to \cdots T_n \to a} \cap \xFs$ then
\begin{itemize}
\item \smallskip
$\ar{f} = n$ implies $a \geqdot_\xB T_i$ and
$\spos{a}{T_i} = \varnothing$ for all $1 \leq i \leq n$
where $\spos{a}{\cdot}$ is defined in \appref{spos},
\item \smallskip
$\ar{f} < n$ implies $\acc{f} = \varnothing$
as well as $a \geqdot_\xB T_i$ and
$T_{\ar{f}+1} \to \dots \to T_n \to a \geqdot T_i$ for all
$1 \leq i \leq \ar{f}$.
\end{itemize}

We are now able to lift CPO with accessible subterms and small
symbols \cite{BJR15} to an appropriate component of a
$\beta\eta$-normal higher-order reduction order.

\begin{definition}
Given a finite set $X$ of variables, the order
$>^X \subseteq \nfbe \times \nfbe$ is inductively defined in
\figref{ncpo} where we implicitly assume $s \in \Lambda_\nv$ whenever $s >^X t$.
Furthermore, $s >_\tau^X t$ if $s >^X t$
and $\tau(s) \succeq_\xT \tau(t)$, and $>$ \textup{($>_\tau$)}
is a shorthand for $>^\varnothing$
\textup{($>_\tau^\varnothing$)}.
\end{definition}

\begin{figure}[t]
\centering
\begin{tabular}{r@{~~}l}
\C{\xFb{\subt}} &
$f(\bar{t}) >^X v$ if $f \in \xFb$ and
$t_i \bsubteq \cdot \asubteq \cdot \geq_\tau v$ for some $i$
\smallskip \\
\C{\xFb{=}} &
$f(\bar{t}) >^X g(\bar{u})$ if $f \in \xFb$,
$f \simeq_\xF g$, $f(\bar{t}) >^X u_i$ for all $i$ and \\
& \quad $\bar{t} \mathrel{({>_\tau} \cup
{\structsm{X} \cdot \geq_\tau})_{\stat{f}}} \bar{u}$
\smallskip \\
\C{\xFb{\succ}} &
$f(\bar{t}) >^X g(\bar{u})$ if $f \in \xFb$,
$f \succ_\xF g$ and $f(\bar{t}) >^X u_i$ for all $i$
\smallskip \\
\C{\xFb@} &
$f(\bar{t}) >^X uv$ if $f \in \xFb$,
$f(\bar{t}) >^X u$ and $f(\bar{t}) >^X v$
\smallskip \\
\C{\xFb\lambda} &
$f(\bar{t}) >^X \lambda y.v$ if $f \in \xFb$,
$f(\bar{t}) >^{X \cup \{z\}} v[y/z]$, $\tau(y) = \tau(z)$ and
$z$ fresh
\smallskip \\
\C{\xFb\xV} &
$f(\bar{t}) >^X y$ if $f \in \xFb$ and $y \in X$
\smallskip \\
\C{@{\subt}} &
$tu >^X v$ if $t \geq^X v$ or $u \geq_\tau^X v$
\smallskip \\
\C{@{=}} &
$tu >^X t'u'$ if $t = t'$ and $u >^X u'$ or $tu >_@^X t'$ and
$tu >_@^X u'$ where \\
& \quad $tu >_@^X v$ if $t >_\tau^X v$ or $u \geq_\tau^X v$ or
$tu >_\tau^X v$
\smallskip \\
\C{@\lambda} &
$tu >^X \lambda y.v$ if $tu >^X v[y/z]$, $\tau(y) = \tau(z)$ and
$z$ fresh
\smallskip \\
\C{@\xFs} &
$tu >^X f(\bar{v})$ if $f \in \xFs$ and
$t u >_\tau^X v_i$ for all $i$
\smallskip \\
\C{@\xV} &
$tu >^X y$ if $y \in X$
\smallskip \\
\C{\lambda{\subt}} &
$\lambda x.t >^X v$ if $t[x/z] \geq_\tau^X v$, $\tau(x) = \tau(z)$ and
$z$ fresh
\smallskip \\
\C{\lambda{\subt}\eta} &
$\lambda x.t >^X v$ if $t[x/z] \geq_\tau^X vz$, $\tau(x) = \tau(z)$ and
$z$ fresh
\smallskip \\
\C{\lambda{=}} &
$\lambda x.t >^X \lambda y.v$ if $t[x/z] >^X v[y/z]$,
$\tau(x) = \tau(y) = \tau(z)$ and $z$ fresh
\smallskip \\
\C{\lambda{\neq}} &
$\lambda x.t >^X \lambda y.v$ if $\lambda x.t >^X v[y/z]$,
$\tau(x) \neq \tau(y) = \tau(z)$ and $z$ fresh
\smallskip \\
\C{\lambda\xFs} &
$\lambda x.t >^X f(\bar{v})$ if $f \in \xFs$ and
$\lambda x.t >_\tau^X v_i$ for all $i$
\smallskip \\
\C{\lambda\xV} &
$\lambda x.t >^X y$ if $y \in X$
\smallskip \\
\C{\xFs{\subt}} &
$f(\bar{t}) >^X v$ if $f \in \xFs$ and $t_i \geq_\tau v$
for some $i$
\smallskip \\
\C{\xFs{=}} & 
$f(\bar{t}) >^X g(\bar{u})$ if $f \in \xFs$,
$f \simeq_\xF g$,
$f(\bar{t}) >_\tau^X u_i$ for all $i$ and
$\bar{t} \mathrel{(>_\tau)_{\stat{f}}} \bar{u}$
\smallskip \\
\C{\xFs{\succ}} &
$f(\bar{t}) >^X g(\bar{u})$ if
$f \in \xFs$, $f \succ_\xF g$ and
$f(\bar{t}) >_\tau^X u_i$ for all $i$
\smallskip \\
\C{\xFs@} &
$f(\bar{t}) >^X uv$ if $f \in \xFs$,
$f(\bar{t}) >_\tau^X u$ and $f(\bar{t}) >_\tau^X v$
\smallskip \\
\C{\xFs\xV} & $f(\bar{t}) >^X y$ if $f \in \xFs$ and $y \in X$
\end{tabular}
\caption{Rules of NCPO.}
\label{fig:ncpo}
\end{figure}

Note that ${>^X}$ is well-defined by induction on the tuple $(s,t)$
with respect to the well-founded relation $({\subt},{\subt})_\lex$. It
is defined like CPO with accessible subterms and small symbols
\cite{BJR15} but with the restriction to terms in $\beta\eta$-normal
form where we additionally require that the first argument of $>^X$
as well as the first argument and all intermediate terms in the
recursive definition of the subterm relation $\bsubt$ are
nonversatile. Furthermore, the rules \C{@\beta} and \C{\lambda\eta}
which orient $\beta$- and $\eta$-reduction, respectively, have been
removed. We also weakened the rule \C{\xFs{=}} by disallowing
$\structsm{X} \cdot \geq_\tau$ in the comparison of the arguments as
used in \C{\xFb{=}}. Finally, we added the rule \C{\lambda{\subt}\eta}
which is inspired by rule (11) from NHORPO \cite{JR15}. We refer to
$>_\tau$ as the $\beta\eta$-normal computability path order
\textup{(NCPO)} and use the symbol $\sqsupset$ with the same
decorations as $>$ to denote \textup{CPO} with accessible subterms
and small symbols.
We illustrate the definition on a number of examples and emphasize
some differences between NCPO and NHORPO together with the
transformation technique of neutralization. For our purposes, it is
sufficient to know that the definition of NHORPO is similar to the
one of NCPO without the extensions of accessible subterms and small
symbols but with fewer rules, without a set $X$ of variables and a
weak decrease in the admissible type order in each recursive
invocation. Furthermore, the transformation technique of
neutralization maps terms to terms with the same termination behavior.
The goal is to transform terms which are beyond the scope of NHORPO
into terms which NHORPO can deal with.
In particular, the transformation may apply a function symbol's
arguments of nonbase type to special terms and apply $\beta$-reduction
in order to minimize the number of lambda abstractions.
A detailed definition can be found in \cite{JR15}.

\begin{example}
\label{exa:diff}
In this first example, we use the rule \C{\lambda{\subt}\eta} which is
not part of CPO to prove termination of the symbolic differentiation
example from \cite{JR15}.
Let $\xB = \SET{\m{r}}$ and consider the function symbols
$\m{sin}, \m{cos} \in \xF_{\m{r} \to \m{r}}$,
$\m{diff} \in \xF_{(\m{r} \to \m{r}) \to \m{r} \to \m{r}}$ and
${+}, {\times} \in
\xF_{(\m{r} \to \m{r}) \to (\m{r} \to \m{r}) \to \m{r} \to \m{r}}$ as
well as the variables $x \in \xV_{\m{r}}$ and
$F, G \in \xF_{\m{r} \to \m{r}}$. Furthermore,
$\ar{\m{sin}} = \ar{\m{cos}} = \ar{\m{diff}} = 1$ and
$\ar{+} = \ar{\times} = 2$. We will use infix notation for $+$ and
$\times$. The HRS $\xR$ defining the symbolic differentiation of
$\m{sin}$ and $\times$ is defined by the following two rules:
\begin{align*}
\m{diff}(\lambda x.\m{sin}(F x))
&\R (\lambda x.\m{cos}(F x)) \times \m{diff}(F) \\
\m{diff}(F \times G) &\R (\m{diff}(F) \times G) + (F \times \m{diff}(G))
\end{align*}
For the termination proof with
NCPO, all function symbols can be big with multiset status and we
will use the precedence
$\m{diff} \succ_\xF \m{sin}, \m{cos}, {+}, {\times}$. Note that all
subterms of left-hand sides except for variables and the
application $F x$ are nonversatile. Since
$\xR$ is an HRS, $\tau(\ell) = \tau(r)$ for all $\ell \R r \in \xR$,
so we only have to check $\ell > r$ for all $\ell \R r \in \xR$. 
For the first rule, we apply
\C{\xFb{\succ}} to get the two proof obligations
$\m{diff}(\lambda x.\m{sin}(F x)) > \lambda x.\m{cos}(F x)$ and
$\m{diff}(\lambda x.\m{sin}(F x)) > \m{diff}(F)$. For the former
one, we proceed by \C{\xFb\lambda}, \C{\xFb{\succ}} and \C{\xFb@}
to obtain the subgoals
$\m{diff}(\lambda x.\m{sin}(F x)) >^{\{z\}} F$ and
$\m{diff}(\lambda x.\m{sin}(F x)) >^{\{z\}} z$. We can directly
resolve the second subgoal with \C{\xFb\xV}. For the first goal,
applying \C{\xFb{\subt}} and then \C{\lambda{\subt}\eta} yields the
goal $\m{sin}(F y) > F y$ which can be handled by \C{\xFb{\subt}}.
For the remaining proof obligation
$\m{diff}(\lambda x.\m{sin}(F x)) > \m{diff}(F)$ we apply
\C{\xFb{=}} and \C{\lambda{\subt}\eta} to obtain the goal
$\m{sin}(F x) >_\tau F x$ which is again handled by
\C{\xFb{\subt}}.%
\footnote{In \cite{JR15}, the same reasoning is used to
handle this subcase but an application of the rule corresponding
to \C{\lambda{\subt}\eta} is not allowed by their definition as
$F$ is a variable. Furthermore, as already mentioned,
$\lambda x.\m{sin}(F x)$ is nonversatile by our sufficient
condition but not by the one given in \cite{JR15}.}
Now consider
the second rule. From two applications of \C{\xFb{\succ}} we obtain
the proof obligations $\m{diff}(F \times G) > \m{diff}(F)$,
$\m{diff}(F \times G) > G$,
$\m{diff}(F \times G) > F$ and
$\m{diff}(F \times G) > \m{diff}(G)$ which can be resolved by two
applications of \C{\xFb{\subt}} or \C{\xFb{=}} followed by one
application of \C{\xFb{\subt}}, respectively.
\end{example}

Note that the argumentation given for the last rule also works
for NHORPO as can be confirmed by both our implementation of NHORPO as
well as the prototype implementation linked from \cite{JR15}.
Hence, NHORPO does not need neutralization in order to prove
termination of the preceding example even though this is claimed in
\cite{JR15}. However, as already mentioned, the original
implementation of NHORPO cannot orient the first rule as the
used sufficient condition for nonversatility is too weak.
Interestingly, neutralization can make up for that, so the proof
checked by the original implementation is correct but for different
reasons than the ones given in \cite{JR15}. The following example
shows that NCPO can prove termination of systems where NHORPO
succeeds only if neutralization is employed.

\begin{example}
\label{exa:nnf}
In this example we use accessible subterms to prove
termination of the computation of negation normal forms of
formulas in first-order logic in the framework of $\beta\eta$-normal
higher-order rewriting using NCPO. We need to distinguish between
terms and formulas, so let the set of base types under consideration
be $\xB = \SET{t,f}$. Furthermore, we consider the logical
connectives represented by the function symbols
${\lnot} \in \xF_{f \to f}$, ${\land}, {\lor} \in \xF_{f \to f \to f}$
and ${\forall}, {\exists} \in \xF_{(t \to f) \to f}$
which are all considered to be big. We set
$\ar{\lnot} = \ar{\forall} = \ar{\exists} = 1$ and
$\ar{\land} = \ar{\lor} = 2$ and allow us to use syntactic sugar in
writing as few parentheses as possible by establishing that the unary
function symbols bind stronger than the binary function symbols.
Furthermore, we use infix notation for $\land$ and $\lor$.
As variables, we will use $P, Q \in \xV_f$ and $R \in \xV_{t \to f}$.
Hence, the HRS $\xR$ for the computation of negational normal forms
in first-order logic is given as the following set of rules:
\begin{align*}
\lnot \lnot P &\R P &
\lnot (P \land Q) &\R \lnot P \lor \lnot Q &
\lnot \forall R &\R \exists (\lambda x.\lnot (R x)) \\
& &
\lnot (P \lor Q) &\R \lnot P \land \lnot Q &
\lnot \exists R &\R \forall (\lambda x.\lnot (R x))
\end{align*}
Note that all non-variable subterms of the left-hand sides are
nonversatile. We choose $f \succ_\xT t$ and let all function
symbols have multiset status. Furthermore, set
${\lnot} \succ_\xF {\land}, {\lor}, {\forall}, {\exists}$. We
have $\lnot \lnot P > P$ by applying \C{\xFb{\subt}} twice.
For $\lnot (P \land Q) > \lnot P \lor \lnot Q$, applying
\C{\xFb{\succ}} yields the subgoals 
$\lnot (P \land Q) > \lnot P$ and
$\lnot (P \land Q) > \lnot Q$ which are handled by \C{\xFb{=}} since
$P \land Q >_\tau P, Q$ by \C{\xFb{\subt}}. We obtain
$\lnot (P \lor Q) > \lnot P \land \lnot Q$ in the same way. This
leaves us with establishing
$\lnot \forall R > \exists (\lambda x.\lnot (R x))$ as the final rule
can again be oriented with the same strategy. By applying
\C{\xFb{\succ}}, \C{\xFb\lambda} and \C{\xFb{=}} we obtain
$\forall R \mathrel{{>_\tau} \cup {\structsm{\SET{x}} \cdot \geq_\tau}}
R x$ as proof obligation. Note that
$\forall R >_\tau R x$ is impossible as $x$ does not occur in the
left-hand side. However, we can establish
$\forall R \structsm{\SET{x}} R x$
which is enough to fulfill the proof obligation. Since
$f \,\geqdot_\xB\, t \to f$ and
$\SET{2} = \posof{f}{t \to f} \subseteq \posp{t \to f} = \SET{2}$ we can
set $1 \in \acc{\forall}$. Hence, $\forall R \asubt R$. Furthermore,
$\tau(\forall R) = \tau(R x) = f$ and $\posof{f}{t} = \varnothing$,
so $\forall R \structsm{\SET{x}} R x$ as desired.
\end{example}

\begin{example}
\label{exa:mapInc}
In this example, we demonstrate the usefulness of small function
symbols by proving termination of the map function together with a
function that increments lists of natural numbers in successor
notation. This is a slightly modified
subsystem of \texttt{AotoYamada\_05\_\_014}
from the termination problem database.%
\footnote{\url{https://github.com/TermCOMP/TPDB}}
Proving termination of the full system for plain higher-order rewriting
needs the transitive closure of CPO as shown in
\cite[Example 8.20]{BJR15}, but this does not help in our setting as the
``middle term'' is not in $\beta\eta$-normal form.%
\footnote{Actually, our modification also has to be applied to
\cite[Example 8.20]{BJR15} because admissible type orders
as defined in \cite{BJR15} do not support equivalent base types.}
Consider the
set of base types $\xB = \SET{\m{a},\m{b}}$ as well as the function
symbols $\m{0} \in \xF_{\m{b}}$, $\m{nil} \in \xF_{\m{a}}$,
$\m{s} \in \xF_{\m{b} \to \m{b}}$,
$\m{plus} \in \xF_{\m{b} \to \m{b} \to \m{b}}$,
$\m{inc} \in \xF_{\m{a} \to \m{a}}$,
$\m{map} \in \xF_{(\m{b} \to \m{b}) \to \m{a} \to \m{a}}$ and
$\m{cons} \in \xF_{\m{b} \to \m{a} \to \m{a}}$. We set
$\ar{\m{0}} = \ar{\m{nil}} = 0$,
$\ar{\m{s}} = \ar{\m{plus}} = \ar{\m{inc}} = 1$ and
$\ar{\m{map}} = \ar{\m{cons}} = 2$. We will use the variables
$x, y \in \xF_{\m{b}}$, $v \in \xF_{\m{a}}$ and
$F \in \xV_{\m{b} \to \m{b}}$. Consider the HRS $\xR$ consisting
of the following rules:
\begin{align*}
\m{plus}(\m{0})\: x &\R x &
\m{plus}(\m{s}(y))\: x &\R \m{s}(\m{plus}(y)\: x) \\
\m{map}(F,\m{nil}) &\R \m{nil} &
\m{map}(F,\m{cons}(x,v)) &\R \m{cons}(F x,\m{map}(F,v)) \\
\m{inc}(v) &\R \m{map}(\m{plus}(\m{s}(\m{0})),v)
\end{align*}
We choose $\m{a} \succ_\xT \m{b}$,
consider all function symbols except for $\m{s}$ to be big and
let them all have multiset status. The used precedence is
$\m{inc} \succ_\xF \m{map} \succ_\xF \m{cons}, \m{nil}$ and
$\m{plus} \succ_\xF \m{s}$. Accessible arguments are not needed.
Note that all non-variable subterms of the left-hand sides are
nonversatile. The first rule follows directly from \C{@{\subt}}.
For the second rule, we apply \C{@\xFs} to obtain the subgoal
$\m{plus}(\m{s}(y))\: x >_\tau \m{plus}(y)\: x$. We proceed by
\C{@{=}} and choose the proof obligations
$\m{plus}(\m{s}(y)) >_\tau \m{plus}(y)$ as well as $x \geq_\tau x$. The
second one is obviously true while the first rule can be resolved
by \C{\xFb{=}} and \C{\xFs{\subt}}. The third rule follows
directly from \C{\xFb{\subt}}. For the fourth rule, we apply
\C{\xFb{\succ}} to obtain the subgoal
$\m{map}(F,\m{cons}(x,v)) > F x$ which follows from \C{\xFb@} and
three applications of \C{\xFb{\subt}} as well as the subgoal
$\m{map}(F,\m{cons}(x,v)) > \m{map}(F,v)$ which follows from
\C{\xFb{=}} and \C{\xFb{\subt}}. Finally, the fifth rule can be
oriented by applying \C{\xFb{\succ}} four times as well as
\C{\xFb{\subt}} once.
\end{example}

The previous example cannot be handled by NHORPO (with or without
neutralization) as the rule for the recursive case of $\m{plus}$
needs small symbols which are neither part of NHORPO nor
added by neutralization.

\section{Correctness Proof}
\label{sec:proof}

We start by a technical result for CPO with accessible subterms
which is needed for establishing $\beta\eta$-normal stability of
$(>_\tau,\sqsupset_\tau)$. In particular, it states an important
connection between $\structsm{X}$ and $\sqsupset^X$.

\begin{lemma}
\label{lem:StructSmImpliesSuccX}
Let $X$ be a finite set of variables and consider the term
$f(\seq{t})$ with $f \in \xFb$. If
$t_i \structsm{X} \cdot \sqsupseteq_\tau v$ for some $1 \leq i \leq n$
then $f(\seq{t}) \sqsupset^X v$.
\end{lemma}

Note that the previous lemma holds for both CPO with accessible
subterms and NCPO. Hence, it allows us to do the usual optimizations
in the implementation of the rule \C{\xFb{=}} given in
\figref{optimizedFEqRule}. This can be justified in both cases as
follows: If $t_i \geq_\tau u_j$ then $f(\bar{t}) >^X u_j$ by
\C{\xFb{\subt}} and if $t_i \structsm{X} \cdot \geq_\tau u_j$ then
$f(\bar{t}) >^X u_j$ by \lemref{StructSmImpliesSuccX}.
The following lemma is the main result needed
for establishing $\beta\eta$-normal stability of the pair
$(>_\tau,\sqsupset_\tau)$.

\begin{figure}[t]
\centering
\begin{tabular}{r@{~~}l}
\C{\xFb{=}\mul} &
$f(\bar{t}) >^X g(\bar{u})$ if $f \simeq_\xF g$,
$\stat{f} = \mul$ and $\bar{t}
\mathrel{({>_\tau} \cup {\structsm{X} \cdot \geq_\tau})_{\mul}}
\bar{u}$
\smallskip \\
\C{\xFb{=}\lex} &
$f(\bar{t}) >^X g(\bar{u})$ if $f \simeq_\xF g$,
$\stat{f} = \lex$ and \\
& $\exists\h i.\,t_i
\mathrel{{>_\tau} \cup {\structsm{X} \cdot \geq_\tau}} u_i$
and $\forall j < i.\,t_j = u_j$
and $\forall j > i.\,f(\bar{t}) >^X t_j$
\end{tabular}
\caption{Optimized Rules for the Implementation of \C{\xFb{=}}.}
\label{fig:optimizedFEqRule}
\end{figure}

\begin{lemma}
\label{lem:subst}
Let $s >^X t$. For every $\beta\eta$-normal substitution $\sigma$
away from $X$ we have $s\sigma{\D} \sqsupset^X t'$ for some
$t\sigma \Rsbe t' \Rnbe t\sigma{\D}$.
Additionally, if $t = uv$ then $t' = u'v'$ with
$u\sigma \Rsbe u' \Rnbe u\sigma{\D}$ and
$v\sigma \Rsbe v' \Rnbe v\sigma{\D}$. 
\end{lemma}

\begin{theorem}
\label{thm:order}
The pair $(>_\tau,\sqsupset_\tau)$ is a $\beta\eta$-normal higher-order
reduction order.
\end{theorem}

\begin{proof}
Since $\sqsupset_\tau$ contains $\beta\eta$-reduction by definition,
\cite[Lemma 8.2]{BJR15} establishes its monotonicity and
\cite[Theorem 8.14]{BJR15} its well-foundedness,
we are left to prove $\beta\eta$-normal stability. Let $s >_\tau t$ and
$\sigma$ be a
substitution. By \lemref{subst} we obtain $s\sigma{\D} \sqsupset t'$
for some $t' \Rnbe t\sigma{\D}$. Since types are preserved under
the application of substitutions as well as $\beta\eta$-conversions
between well-typed terms, we obtain
$s\sigma{\D} \sqsupset_\tau t'$. Finally, we obtain
$s\sigma{\D} \sqsupset_\tau^+ t\sigma{\D}$ as
$\beta\eta$-steps are included in $\sqsupset_\tau$. \qed
\end{proof}

\section{Implementation}
\label{sec:nhorpo}

A prototype implementation of NCPO is available at GitHub.%
\footnote{\url{https://github.com/niedjoh/hrsterm}}
In contrast
to the implementations of NHORPO with neutralization as well as CPO
linked from \cite{JR15} and \cite{BJR15}, respectively, our implementation
is not a mere termination checker which requires all parameters of
the order as input but searches for suitable parameters using
SAT/SMT which is standard practice for termination tools such as
\cite{K20,G+17}. Implementing the search for parameters
also allows one to quickly find mistakes like in
\cite[Example 8.19]{BJR15} where it is claimed that small symbols are
needed whereas this is not true. Running the prototype
implementation linked from \cite{BJR15} with the parameters generated
from our prototype implementation confirms our findings.

Our prototype is implemented in Haskell and uses the recent Hasmtlib
package%
\footnote{\url{https://github.com/bruderj15/Hasmtlib}} to encode the
termination problems into the SMT-LIB 2 format%
\footnote{\url{https://smt-lib.org/}} and communicate with SMT solvers
which have to be installed separately. We use the well-known TPTP THF
format \cite{S17} as the input format of our prototype
implementation. Our parser supports a suitable fragment of TPTP THF
which consists of unit clauses where equality is the only allowed
predicate. Variables in rules are handled by universal
quantification. In particular, the chosen input format does not
support applied function symbols with fixed arities but solely relies
on application. We think this is appropriate as in higher-order
problems, fixed arities are a superfluous limitation in general and
only needed as additional structure for termination techniques based
on path orders. Hence, our implementation transforms the input into
a representation using applied function symbols with fixed arities for
the termination proof. The transformation chooses the maximal possible
arity for each function symbol which depends on the minimal number of
arguments which it is applied to in the problem. This gives us as much
structure as possible without having to write the arities down
explicitly.

The SMT encoding of NCPO contains more recursive cases and more global
conditions than the one for NHORPO, but except for mild cases like
$\structsm{X}$, terms only need to be decomposed. In particular,
we can easily encode the search for precedences and admissible type
orders of the class given in \lemref{admissibleOrdering} by
assigning an integer variable to each function symbol and base
type. The remaining parameters regarding the status of each function
symbol, which base types are basic, accessible arguments and small
symbols can be encoded by boolean variables. Apart from some global
conditions, the encoding then just needs to assert orientability of
each rule in a given HRS by taking the disjunction of all applicable
cases in the definition of NCPO and proceed recursively.
Unlike stated in \cite{JR15}, this drastically changes for the
transformation technique of normalization. In an efficient encoding,
one would have to syntactically analyze $\beta\eta$-normal forms of
transformed terms
of the original system depending on parameters which are only given
symbolically as the goal of the encoding is to search for suitable
values for them. This seems to be impossible unless the concrete
values of these parameters are hard-coded in a big
disjunction which makes the encoding much more verbose than
the usual encodings of precedences and the like. In our reimplementation
of NHORPO and neutralization, we precomputed the terms resulting
from the concrete parameters of neutralization in Haskell and used
the SMT solver interactively to try to find a solution for each
encoding in turn.

Apart from our reimplementation of NHORPO, we know of two other
fully automated implementations of HORPO which are applicable to
higher-order rewriting modulo $\beta\eta$, WANDA~\cite{K20} and
\csiho~\cite{N17}. WANDA implements a variation of HORPO for the
different rewrite formalism of AFSMs. AFSMs can model our
$\beta\eta$-normal flavor of HRSs by replacing the free variables
with meta-variables, their applications by meta-applications and
employing a $\beta\eta$-first reduction strategy which rewrites a
term to $\beta\eta$-normal form before performing a rewrite
step. However, a formal proof of such a result is only
available for a transformation of the subclass of pattern
HRSs in $\beta\eta$-long normal form \cite[Transformation 3.4]{K12}
to AFSMs. Furthermore, WANDA only supports a $\beta$-first
strategy, so in general the HORPO implementation in WANDA and our
implementation of NHORPO/NCPO cannot be directly compared. \csiho
is a confluence tool which implements a basic variant of HORPO which
can be used for pattern HRSs \cite{vR01} in $\beta\eta$-long normal
form. Hence, it is also not possible to directly compare the HORPO
implementation of \csiho with our prototype. Nevertheless, a
comparison with both WANDA and \csiho is instructive as the
investigated termination problems only differ in the treatment of
$\eta$.

\begin{table}[t]
\centering
\caption{Experimental results.}
\label{tab:experiments}
\newcommand{\ptime}[1]{\phantom{33.968~s}\llap{#1}}
\begin{tabular}{@{}l@{\qquad}c@{\qquad}c@{\quad}c@{}}
\toprule
problem & NCPO & NHORPO & NHORPO+neutralization \\
\midrule
\exaref{diff} (extended) &
$\checkmark$ \quad 0.080~s & $\checkmark$ \quad 0.022~s &
$\checkmark$ \quad \ptime{0.410~s} \\
\exaref{nnf} &
$\checkmark$ \quad 0.043~s & $\times$ \quad 0.011~s &
$\checkmark$ \quad \ptime{2.286~s} \\
\exaref{mapInc} &
$\checkmark$ \quad 0.020~s & $\times$ \quad 0.010~s &
$\times$ \quad \ptime{0.322~s} \\
\cite[Example 5.2]{BJR15} &
$\checkmark$ \quad 0.015~s & $\times$ \quad 0.009~s &
$\checkmark$ \quad \ptime{0.009~s} \\
\cite[Example 7.1]{BJR15} &
$\checkmark$ \quad 0.021~s & $\times$ \quad 0.011~s &
$\times$ \quad \ptime{33.968~s} \\
\cite[Example 8.19]{BJR15} &
$\checkmark$ \quad 0.026~s & $\times$ \quad 0.011~s &
$\times$ \quad \ptime{12.408~s} \\
\cite[Example 7.1]{JR15} &
$\checkmark$ \quad 0.032~s & $\times$ \quad 0.013~s &
$\checkmark$ \quad \ptime{0.498~s} \\
\cite[Example 7.2]{JR15} &
$\times$ \quad 0.016~s & $\times$ \quad 0.010~s &
$\times$ \quad \ptime{0.340~s} \\
\cite[Example 7.3]{JR15} &
$\checkmark$ \quad 0.025~s & $\times$ \quad 0.011~s &
$\checkmark$ \quad \ptime{0.166~s} \\
\texttt{neutr.p} &
$\times$ \quad 0.012~s & $\times$ \quad 0.008~s &
$\checkmark$ \quad \ptime{0.025~s} \\
\texttt{neutrN.p} &
$\checkmark$ \quad 0.018~s & $\checkmark$ \quad 0.010~s &
$\checkmark$ \quad \ptime{0.030~s} \\
\bottomrule
\end{tabular}
\end{table}

\tabref{experiments} compares our implementation of NCPO with our
reimplementation of NHORPO with and without neutralization on a small
set of problems. In all configurations, termination ($\checkmark$) or
the unsatisfiability of the encoding ($\times$) are checked using the
SMT solver Z3 \cite{dMB08}. Here, unsatisfiability of the encoding
means that termination cannot be established with the given method,
and does not imply nontermination. Furthermore, the table contains the
execution time of each invocation on an Intel Core i7-7500U CPU
running at a clock rate of 2.7 GHz with 15.5 GiB of main memory. As
has already been pointed out in \cite{JR15}, NHORPO on its own is
quite weak, so it is intended to be used together with
neutralization. For four problems, neutralization on top of NHORPO is
not enough: As already discussed, \exaref{mapInc} needs small symbols
and \cite[Example 8.19]{BJR15} does not work because (N)HORPO enforces
all subgoals to be weakly oriented in the admissible order on types
while this is not the case for (N)CPO. Furthermore, \cite[Example
7.1]{BJR15} shows that neutralization does not subsume accessible
subterms. Finally, note that \cite[Example 7.2]{JR15} is not solvable
by any of the three methods. In \cite{JR15} it is claimed that NHORPO
with neutralization is able to prove its termination but the proof
given there does not work since it contains comparisons where the
left-hand side is versatile. Indeed, our experiments confirm that
there does not exist a suitable parameter assignment for
neutralization which enables a termination proof with NHORPO. The
existence of \texttt{neutr.p} shows that NCPO and NHORPO with
neutralization have incomparable power. Note that NCPO succeeds in
establishing termination of a neutralized version (\texttt{neutrN.p})
of this problem. With respect to execution time, it can be seen that a
search for suitable parameters of NCPO does not take much more time
than the corresponding search for NHORPO parameters. For
neutralization, the number of parameters and therefore the execution
time of our search method can explode quickly in the presence of
function symbols with big arities which have more than one argument of
nonbase type. However, there may be a more efficient encoding of the
search for the neutralization parameters than the one implemented by
us. We also evaluated the HORPO implementations in WANDA and
\csiho on the problems given in \tabref{experiments}. The HORPO
implementation in WANDA yields termination proofs for
\texttt{neutr.p} and \texttt{neutrN.p} which shows that it is
incomparable with NHORPO. The HORPO implementation in \csiho performs
worse than NHORPO which is to be expected because it is based on a very
basic version of HORPO. The examples as well as a script to reproduce
our experiments are available on the GitHub repository of our
prototype implementation.

\section{Conclusion}
\label{sec:conclusion}

In this paper, we lifted the computability path order with its
extensions for accessible subterms and small symbols \cite{BJR15} from
plain to $\beta\eta$-normal higher-order rewriting. In order to
achieve that goal, we followed the approach from \cite{JR15} based on
$\beta\eta$-normal stability which allowed us to use monotonicity
and well-foundedness of CPO instead of proving the corresponding
properties once more for the new order. In
addition, we gave an improved sufficient condition for terms to be
nonversatile which is an important ingredient of the used
approach. The resulting order, dubbed NCPO, can prove termination of
systems which are out of reach for all other HORPO implementations
targeting normal higher-order rewriting known to us. Moreover, NCPO and
NHORPO with neutralization, its strongest competitor, are of
incomparable power. Contrary to what is claimed in \cite{JR15}, we
find it much more difficult to encode the search for suitable
parameters of NHORPO with neutralization than for NCPO. Thus, we think
that NCPO is a powerful and lightweight alternative to NHORPO with
neutralization.

As far as future work is concerned, transitivity of NCPO is an open
question. For NHORPO and CPO, the inclusion of $\beta$-reduction in
the order justifies that it is not transitive, but this is not the
case for NCPO. If $>_\tau$ were not transitive, then merely
checking $\ell >_\tau r$ for each rule would not be a complete
method for establishing termination with $(>_\tau,\sqsupset_\tau)$
using \thmref{termination}. Furthermore, the logical next step
towards a powerful termination technique for normal higher-order
rewriting would be the integration of NCPO into a dependency pair
technique for normal higher-order rewriting as defined for example in
\cite{KISB09}. Finally, given the complex nature and wide-ranging
pitfalls inherent in higher-order termination as well as the mistakes
found especially in \cite{JR15}, formalizations of higher-order
termination methods are needed. To the best of our knowledge, a
formalization of a basic variant of HORPO without computability
closure \cite{K09} is the only existing work in that direction.

\begin{credits}
\subsubsection{\ackname}
We thank the anonymous reviewers for their valuable comments
which improved the presentation of the paper.

\subsubsection{\discintname}
The authors have no competing interests to declare that are
relevant to the content of this article.
\end{credits}

\bibliographystyle{splncs04}
\bibliography{references}

\appendix

\section{Definition of $\spos{a}{\cdot}$}\label{app:spos}

The following definition of the set $\spos{a}{\cdot}$ is defined
by mutual recursion with four other sets of positions. Each set
corresponds to one computability property used in the termination
proof of CPO with small symbols in \cite[Section 8.3]{BJR15}. This
is needed in order to break a circularity in the inter-dependence of
these computability properties.

\begin{definition}[from \cite{BJR15}]
For each computability property in the termination proof of CPO
with small symbols, we define a set of positions with respect to a
given base type $a$. Here, \textsf{S} stands for the property
(comp-sn), \textsf{R} for (comp-red), \textsf{N} for (comp-neutral),
\textsf{L} for (comp-lam) and \textsf{C} for (comp-small).
\begin{align*}
\cpos{a}{a} &= \SET{\epsilon} \quad \cpos{a}{b} = \varnothing ~
\text{if $a \neq b$} \\
\spos{a}{b} &= \rpos{a}{b} = \npos{a}{b} = \lpos{a}{b} = \varnothing \\
\cpos{a}{U \to V} &= \npos{a}{U \to V} \\
\spos{a}{U \to V} &=
\SET{1p \mid p \in \npos{a}{U}} \cup \SET{2p \mid p \in \spos{a}{V}} \\
\rpos{a}{U \to V} &=
\SET{1p \mid p \in \npos{a}{U}} \cup \SET{2p \mid p \in \spos{a}{V}} \\
\npos{a}{U \to V} &= \SET{1p \mid p \in \spos{a}{U}} \cup
\SET{2p \mid p \in \lpos{a}{V} \cup \cpos{a}{V}} \\
\lpos{a}{U \to V} &= \cpos{a}{U \to V} \cup
\SET{1p \mid p \in \spos{a}{U} \cup \npos{a}{U}} \\
&\phantom{{} = \cpos{a}{U \to V} {}} \cup
\SET{2p \mid p \in \lpos{a}{V} \cup \cpos{a}{V}}
\end{align*}
\end{definition}

\section{Remaining Proofs}

\begin{proof}[of \lemref{nonversatile}]
We prove each case separately.
\begin{enumerate}[(i)]
\item
Direct consequence of the definition.
\item
If $uv \in \nfbe$ then $u$ is not an abstraction. Since $u$ is
nonversatile, for every substitution $\sigma$, $u\sigma{\D}$ is
also not an abstraction. Hence,
$(uv)\sigma{\D} = u\sigma{\D}v\sigma{\D}$
as desired.
\item
Consider $\lambda x.ux$ and an arbitrary substitution $\sigma$. Without
loss of generality, we may assume that $\sigma$ is away from $\SET{x}$
because capture-avoiding substitution is employed. Since
$\lambda x.ux \in \nfbe$, $x \in \fv{u}$. From nonversatility of $ux$ we
obtain $(ux)\sigma{\D} = u\sigma{\D} x$, so we have to
prove that $x \in \fv{u\sigma{\D}}$. To that end, we prove
by induction on $t$ that if $x \in \fv{t}$ and $t$ is nonversatile
in the case it is an abstraction then $x \in \fv{t\sigma{\D}}$.
If $t = x$, we are done immediately. The case $v = y \neq x$ is vacuously
true. If $v = f(\bar{t})$ then the induction hypothesis yields
$x \in \fv{t_i\sigma{\D}}$ whenever $x \in \fv{t_i}$.
Since $x \in \fv{f(\bar{t})}$ we obtain
$x \in \fv{f(\bar{t})\sigma{\D}} =
\fv{f(\bar{t}\sigma{\D})}$ as desired. If $t = uv$
then $t$ is nonversatile by assumption. We have
$x \in \fv{(uv)\sigma{\D}} = \fv{u\sigma{\D}} \cup
\fv{v\sigma{\D}}$ by the induction hypothesis. Finally,
consider $t = \lambda y.u$. Without loss of generality we may
assume that $y \neq x$ and $\sigma$ is away from $\SET{y}$. Hence,
$x \in \fv{(\lambda y.u)\sigma{\D}} =
\fv{u\sigma{\D}} \setminus \SET{y}$ by the induction
hypothesis.
\item
Consider $\lambda x.u$ and an arbitrary substitution $\sigma$.
Again, we may assume that $\sigma$ is away from $\SET{x}$ without
loss of generality. Note that $\lambda x.u$ is nonversatile if
and only if $u\sigma{\D} \neq u'x$ for some $u'$ with
$x \notin \fv{u'}$. If $u = x$ or $u = y \neq x$ we are done
immediately as $\sigma$ is away from $x$. If $u = f(\bar{t})$
then $u\sigma{\D} = f(\bar{t}\sigma{\D})$. If
$u = vw$ it is nonversatile by assumption. Hence,
$(vw)\sigma{\D} = v\sigma{\D} w\sigma{\D}$. We
have to prove
$w\sigma{\D} \neq x$ for the case $w \neq x$. We proceed
by case analysis on $w$. If $w = y \neq x$ we are done
immediately as $\sigma$ is away from $\SET{x}$. If
$w = f(\bar{t})$ then
$w\sigma{\D} = f(\bar{t}\sigma{\D})$. If
$w = w_1w_2$ then $(w_1w_2)\sigma{\D} =
w_1\sigma{\D} w_2\sigma{\D} \neq x$ as $w_1w_2$ is
nonversatile. If $w = \lambda y.w'$ we can assume without loss of
generality that $\sigma$ is away from $\SET{y}$ and
$(\lambda y.w')\sigma{\D} = \lambda y.w'\sigma{\D} \neq x$
as $\lambda y.w'$ is nonversatile by assumption. Hence,
$w\sigma{\D} \neq x$ in any case. Finally, consider
$u = \lambda y.u$. Without loss of generality, $\sigma$ is away
from $\SET{y}$ and from the assumption that $u$ is nonversatile we
obtain $u\sigma{\D} = \lambda y.u\sigma{\D}$. \qed
\end{enumerate}
\end{proof}

\begin{proof}[of \lemref{subtNormStability}]
Let $\sigma$ be a substitution. Clearly, $s = t$ implies
$s\sigma{\D} = t\sigma{\D}$ (\textasteriskcentered).
\begin{enumerate}[(i)]
\item
We proceed by induction on $s \bsubt t$. If $s = f(\seq{s})$ then
there exists an $1 \leq i \leq n$ such that $s_i \bsubteq t$.
The induction hypothesis together with (\textasteriskcentered)
yields
$f(\seq{s})\sigma{\D} =
f(s_1\sigma{\D},\dots,s_n\sigma{\D}) \bsubt
s_i\sigma{\D} \bsubteq t\sigma{\D}$. If
$s = uv$ then $w \bsubteq t$ for some $w \in \SET{u,v}$. We obtain
$(uv)\sigma{\D} = u\sigma{\D}v\sigma{\D} \bsubt
w\sigma{\D} \bsubteq t\sigma{\D}$ by
nonversatility of $s$, the induction hypothesis and
(\textasteriskcentered). If $s = \lambda x.u$ then $u \bsubteq t$.
Without loss of generality, we may assume that $\sigma$ is
away from $\SET{x}$. Hence, we have
$(\lambda x.u)\sigma{\D} = \lambda x.u\sigma{\D} \bsubt
u\sigma{\D} \bsubteq t\sigma{\D}$ by nonversatility of
$s$, the induction hypothesis and (\textasteriskcentered).
\item
We proceed by induction on $s \asubt t$. Let
$s = f(s_1,\dots,s_{\ar{f}}) s_{\ar{f}+1} \cdots s_n$
and $s_i \asubteq t$ for some $i \in \acc{f}$. The induction
hypothesis together with (\textasteriskcentered) yields
$s_i\sigma{\D} \asubteq t\sigma{\D}$. From
$s\sigma{\D} =
f(s_1\sigma{\D},\dots,s_{\ar{f}}\sigma{\D})
s_{\ar{f}+1}\sigma{\D} \cdots s_n\sigma{\D}$ we obtain
$s\sigma{\D} \asubt t\sigma{\D}$ as desired.
\item
Let $s \structsm{X} t$ and $\sigma$ be a substitution away from $X$.
By definition of $\structsm{X}$, $t = u x_1 \cdots x_k$ for some
$\seq[k]{x} \in X$ and $u$ with $s \asubt u$. From part (ii) we obtain
$s\sigma{\D} \asubt u\sigma{\D}$. Since $\sigma$
is away from $X$, we obtain
$s\sigma{\D} \structsm{X} u\sigma{\D} x_1 \cdots x_k
\Rnbe (u x_1 \cdots x_n)\sigma{\D}$ as desired.
\item
Let $s \structsm{X} t$, $t \in \Lambda_\nv$ and $\sigma$ be a substitution
away from $X$. By definition of $\structsm{X}$, $t = u x_1 \cdots x_k$ for
some $\seq[k]{x} \in X$ and $u$ with $s \asubt u$. We obtain
$s\sigma{\D} \structsm{X} u\sigma{\D} x_1 \cdots x_k$
by part (iii). From $t \in \Lambda_\nv$ we conclude that
$u \in \Lambda_\nv$ is not an abstraction. Hence,
$u\sigma{\D}$ is also not an abstraction and
$t\sigma{\D} = u\sigma{\D} x_1 \cdots x_k$ as
desired. \qed
\end{enumerate}
\end{proof}

\begin{proof}[of \lemref{StructSmImpliesSuccX}]
In the following, we denote $(\seq{t})$ by $\bar{t}$.
Let $t_i \structsm{X} \cdot \sqsupseteq_\tau v$, so there exists a term
$u$ such that $t_i \asubt u$ and $u x_1 \cdots x_k \sqsupseteq_\tau v$ for
some $\seq[k]{x} \in X$. We proceed by induction on $k$ for
arbitrary $v$ (1). If $k = 0$ then $u \sqsupseteq_\tau v$ and therefore
$f(\bar{t}) \sqsupset^X v$ by \C{\xFb{\subt}}. For the step case, we
first
consider $u x_1 \cdots x_{k+1} = v$. Then, $k+1$ applications of
\C{\xFb@} and \C{\xFb\xV} together with one application of
\C{\xFb{\subt}} establish $f(\bar{t}) \sqsupset^X v$. We proceed by
inner induction on $u x_1 \cdots x_{k+1} \sqsupset_\tau v$ (2).
\begin{itemize}
\item[\C{@{\subt}}]
Let $u x_1 \cdots x_{k+1} \sqsupset v$ because
$u x_1 \cdots x_k \sqsupseteq v$ or $x_{k+1} \sqsupseteq_\tau v$. In the
latter
case we have $v = x_{k+1}$ and $f(\bar{t}) \sqsupset^X v$ by
\C{\xFb\xV}. In the former case, if $v = u x_1 \cdots x_k$ then $k$
applications of \C{\xFb@} and \C{\xFb\xV} together with one
application of \C{\xFb{\subt}} establish
$f(\bar{t}) \sqsupset^X v$. Otherwise, induction hypothesis (1)
yields $f(\bar{t}) \sqsupset^X v$ as desired.
\item[\C{@{=}}]
Let $u x_1 \cdots x_{k+1} \sqsupset v_1v_2$ because
$(u x_1 \cdots x_k) = v_1$ and $x_{k+1} \sqsupset v_2$. This case is
impossible. Now let $u x_1 \cdots x_{k+1} \sqsupset v_1v_2$ because for
each $w \in \SET{v_1,v_2}$ at least one of
$u x_1 \cdots x_k \sqsupset_\tau w$ or $x_{k+1} \sqsupseteq_\tau w$ or
$u x_1 \cdots x_{k+1} \sqsupset_\tau w$ holds. In the first case,
induction hypothesis (1) yields $f(\bar{t}) \sqsupset^X w$. The second
case demands $w = x_{k+1}$, so $f(\bar{t}) \sqsupset^X w$ by
\C{\xFb\xV}. For the third case, induction hypothesis (2) yields
$f(\bar{t}) \sqsupset^X w$. In any case, we obtain
$f(\bar{t}) \sqsupset^X v_1$ as well as
$f(\bar{t}) \sqsupset^X v_2$ and
therefore $f(\bar{t}) \sqsupset^X v_1v_2$ by \C{\xFb@}.
\item[\C{@\lambda}]
Let $u x_1 \cdots x_{k+1} \sqsupset \lambda y.w$
because $u x_1 \cdots x_{k+1} \sqsupset w[y/z]$ for a fresh variable $z$
with $\tau(z) = \tau(y)$. Induction hypothesis (2) yields
$f(\bar{t}) \sqsupset^X w[y/z]$ and therefore also
$f(\bar{t}) \sqsupset^{X \cup \{z\}} w[y/z]$ since
${\sqsupset^X} \subseteq {\sqsupset^Y}$ whenever $X \subseteq Y$ as can be
easily seen from the definition. Hence,
$f(\bar{t}) \sqsupset^X \lambda y.w$ by \C{\xFb\lambda}.
\item[\C{@\xV}]
This case is impossible as for all $y \in \xV$,
$y \not\in \varnothing$. \qed
\end{itemize}
\end{proof}

\begin{proof}[of \lemref{subst}]
Note that types are preserved under the application of substitutions
as well as $\beta\eta$-conversion between well-typed terms
(\textasteriskcentered). We proceed by induction on $s >^X t$.
\begin{itemize}
\item[\C{\xFb{\subt}}]
Let $f(\seq{t}) >^X v$ because $f \in \xFb$ and
$t_i \bsubteq t_i' \asubteq t_i'' \geq_\tau v$ for some
$1 \leq i \leq n$ and terms $t_i'$, $t_i''$. Note that any term in
this sequence only has to be nonversatile if there is a strict
step after its occurrence. \lemref{subtNormStability}(1,2)
yields $t_i\sigma{\D} \bsubteq t_i'\sigma{\D} \asubteq
t_i''\sigma{\D}$. By the induction hypothesis together with
(\textasteriskcentered) we obtain
$t_i''\sigma{\D} \sqsupseteq_\tau v'$ for some
$v\sigma \Rsbe v' \Rnbe v\sigma{\D}$. Hence, $f(\bar{t})\sigma{\D} =
f(\bar{t}\sigma{\D}) \sqsupset^X v' \Rnbe v\sigma$ by
\C{\xFb{\subt}}.
\item[\C{\xFb{=}}]\smallskip
Let $f(\seq{t}) >^X g(\seq[m]{u})$ because $f \in \xFb$,
$f \simeq g$, $f(\bar{t}) >^X u_i$ for all
$1 \leq i \leq m$ and $(\seq{t})
\mathrel{({>_\tau} \cup {\structsm{X} \cdot \geq_\tau})_{\stat{f}}}
(\seq[m]{u})$. \smallskip

If $\stat{f} = \stat{g} = \lex$ then there exists an
$i < \minrel{n}{m}$ such that
$t_i \mathrel{{>_\tau} \cup {\structsm{X} \cdot \geq_\tau}} u_i$ and
$t_j = u_j$ for all $j < i$. Clearly, we have
$t_j\sigma{\D} = u_j\sigma{\D}$ for all $j < i$. The
induction hypothesis together with \lemref{subtNormStability}(3,4) and
(\textasteriskcentered) yields $t_i\sigma{\D}
\mathrel{{\sqsupset_\tau} \cup {\structsm{X} \cdot \sqsupseteq_\tau}} u_i'$
for some $u_i\sigma \Rsbe u_i' \Rnbe u_i\sigma{\D}$. Furthermore, for all
$m \geq k > i$ there is some $u_k\sigma \Rsbe u_k' \Rnbe u_k\sigma{\D}$
such that $f(\seq{t})\sigma \sqsupset^X u_k'$ by the induction
hypothesis, so overall we obtain
$(t_1\sigma{\D},\dots,t_n\sigma{\D})
\mathrel{({>_\tau} \cup {\structsm{X} \cdot \geq_\tau})_\lex}
(u_1\sigma{\D},\dots,u_i',\dots,u_m')$. For every term
$u \in \SET{u_1\sigma{\D},\dots,u_{i-1}\sigma{\D},u_i'}$,
$f(\seq{t})\sigma{\D} =
f(t_1\sigma{\D},\dots,t_n\sigma{\D}) \sqsupset^X u$ holds by
\C{\xFb{\subt}} or \lemref{StructSmImpliesSuccX}. Hence, we obtain
$f(\seq{t})\sigma{\D} =
f(t_1\sigma{\D},\dots,t_n\sigma{\D}) \sqsupset^X
g(u_1\sigma{\D},\dots,u_i',\dots,u_m') \Rnbe
g(\seq[m]{u})\sigma{\D}$ by \C{\xFb{=}}. \smallskip

If $\stat{f} = \stat{g} = \mul$ then
$\SET{\seq[m]{u}} = (\SET{\seq{t}} \setminus X) \uplus Y$ for some
nonempty set $X \subseteq \SET{\seq{t}}$ and for all $u_i \in Y$ there
exists a $t_j \in X$ with $t_j >^X_\tau u_i$. The induction hypothesis
together with \lemref{subtNormStability}(3,4) and (\textasteriskcentered)
yields $t_j\sigma{\D}
\mathrel{{\sqsupset_\tau} \cup {\structsm{X} \cdot \sqsupseteq_\tau}} u_i'$
for some $u_i\sigma \Rsbe u_i' \Rnbe u_i\sigma$. As observed before,
$u_i = t_j$ implies
$u_i\sigma{\D} = t_j\sigma{\D}$. Hence, we obtain
$(t_1\sigma{\D},\dots,t_n\sigma{\D})
\mathrel{({>_\tau} \cup {\structsm{X} \cdot \geq_\tau})_\mul}
(u_1',\dots,u_m')$ where
$u_i' = u_i\sigma{\D}$ if $u_i \not\in Y$. Moreover,
\C{\xFb{\subt}} or \lemref{StructSmImpliesSuccX} yields
$f(\seq{t})\sigma{\D} =
f(t_1\sigma{\D},\dots,t_n\sigma{\D}) \sqsupset^X u_i'$ for
all $1 \leq i \leq m$. Thus, we obtain $f(\seq{t})\sigma{\D} =
f(t_1\sigma{\D},\dots,t_n\sigma{\D}) \sqsupset^X
g(u_1',\dots,u_m') \Rnbe g(\seq[m]{u})\sigma{\D}$ by
\C{\xFb{=}}.
\item[\C{\xFb{\succ}}] \smallskip
Let $f(\bar{t}) >^X g(\seq[m]{u})$ because $f \in \xFb$,
$f \succ_\xF g$ and $f(\bar{t}) >^X u_i$ for all
$1 \leq i \leq m$. For each $u_i$, the induction hypothesis yields
$f(\bar{t})\sigma{\D} \sqsupset^X u_i'$ for some
$u_i\sigma \Rsbe u_i' \Rnbe u_i\sigma{\D}$ and we obtain
$f(\bar{t})\sigma{\D} = f(\bar{t}\sigma{\D})
\sqsupset^X g(u_1',\dots,u_m') \Rnbe g(\seq[m]{u})\sigma{\D}$
by \C{\xFb{\succ}}.
\item[\C{\xFb@}] \smallskip
Let $f(\bar{t}) >^X uv$ because $f \in \xFb$,
$f(\bar{t}) >^X u$ and $f(\bar{t}) >^X v$. The induction
hypothesis yields $f(\bar{t})\sigma{\D} \sqsupset^X u'$ and
$f(\bar{t})\sigma{\D} \sqsupset^X v'$ for some
$u\sigma \Rsbe u' \Rnbe u\sigma{\D}$ and
$v\sigma \Rsbe v' \Rnbe v\sigma{\D}$. Thus,
$f(\bar{t})\sigma{\D} = f(\bar{t}\sigma{\D})
\sqsupset^X u'v' \Rnbe (uv)\sigma{\D}$ by \C{\xFb@}.
\item[\C{\xFb\lambda}] \smallskip
Let $f(\bar{t}) >^X \lambda y.v$ because $f \in \xFb$,
$f(\bar{t}) >^{X \cup \{z\}} v[y/z]$ for a fresh $z$ with
$\tau(y) = \tau(z)$. Given a substitution $\sigma$ away from $X$, we may
assume that $\sigma$ is away from $\SET{y}$ without loss of generality.
Furthermore, we can always choose a fresh $z'$ with $\tau(z) = \tau(z')$
such that $\sigma$ is away from $\SET{z'}$ and we still have
$f(\bar{t}) >^{X \cup \{z'\}} v[y/z']$. The induction hypothesis
yields $f(\bar{t})\sigma{\D} \sqsupset^{X \cup \{z'\}} v'[y/z']$
for some $v[y/z']\sigma \Rsbe v'[y/z'] \Rnbe v[y/z']\sigma{\D}$.
Note that we
can immediately fix $v'[y/z']$ as the form of the term because
$y$ does not occur freely in $v[y/z']\sigma$ by assumption.
Since $\sigma$ is away from $\SET{y,z'}$
we have $v\sigma \Rsbe v' \Rnbe v\sigma{\D}$ and
$f(\bar{t})\sigma{\D} = f(\bar{t}\sigma{\D})
\sqsupset^X \lambda y.v' \Rsbe \lambda y.v\sigma{\D}
\Rnbe (\lambda y.v)\sigma{\D}$ by \C{\xFb\lambda}.
\item[\C{\xFb\xV}] \smallskip
Let $f(\bar{t}) >^X y$ because $f \in \xFb$ and $y \in X$. Hence,
for $\sigma$ away from $X$ we also have
$f(\bar{t})\sigma{\D} =
f(\bar{t}\sigma{\D}) \sqsupset^X y = y\sigma$ by \C{\xFb\xV}.
\item[\C{@{\subt}}] \smallskip
Let $tu >^X v$ because $t \geq^X v$ or
$u \geq_\tau^X v$. We obtain $t\sigma{\D} \sqsupseteq^X v'$
or $u\sigma{\D} \sqsupseteq_\tau^X v'$ for some
$v\sigma \Rsbe v'\Rnbe v\sigma{\D}$ by the induction hypothesis and
(\textasteriskcentered). Hence, $(tu)\sigma{\D} =
(t\sigma{\D})(u\sigma{\D}) \sqsupset^X v' \Rnbe v\sigma$ by
nonversatility of $tu$ and \C{@{\subt}}.
\item[\C{@{=}}] \smallskip
Let $tu >^X t'u'$ because $t = t'$ and $u >^X u'$.
Then $t\sigma{\D} = t'\sigma{\D}$ and
$u\sigma{\D} \sqsupset^X u''$ for some
$u'\sigma \Rsbe u'' \Rnbe u'\sigma{\D}$ by the induction hypothesis. Hence,
$(tu)\sigma{\D} = (t\sigma{\D})(u\sigma{\D})
\sqsupset^X (t'\sigma{\D})u'' \Rnbe (t'u')\sigma{\D}$ by
nonversatility of $tu$ and \C{@{=}}. Otherwise, let $tu >^X t'u'$
because for each $v \in \SET{t',u'}$ at least one of $t >_\tau^X v$,
$u \geq_\tau^X v$ and $tu >_\tau^X v$ holds. By the induction
hypothesis together with (\textasteriskcentered) there is some
$v\sigma \Rsbe v' \Rnbe v\sigma{\D}$ such that at least one of
$t\sigma{\D} \sqsupset_\tau^X v'$,
$u\sigma{\D} \sqsupseteq_\tau^X v'$ and
$(tu)\sigma{\D} \sqsupset_\tau^X v'$ holds. By denoting $v'$
by $t''$ if $v = t'$ and by $u''$ if $v = u'$, we obtain
$(tu)\sigma{\D} = (t\sigma{\D})(u\sigma{\D})
\sqsupset^X t''u'' \Rnbe (t'u')\sigma{\D}$ by nonversatility of $tu$
and \C{@{=}}.
\item[\C{@\lambda}] \smallskip
The same reasoning as in case \C{\xFb\lambda} applies as $tu$ is
nonversatile.
\item[\C{@\xFs}] \smallskip
Let $tu >^X f(\seq[m]{v})$ because $f \in \xFs$ and
$tu >_\tau^X v_i$ for all $1 \leq i \leq m$. By the induction hypothesis
together with (\textasteriskcentered) we obtain
$(tu)\sigma{\D} \sqsupset_\tau^X v_i'$ for some
$v_i\sigma \Rsbe v_i' \Rnbe v_i\sigma{\D}$ for all $1 \leq i \leq m$.
Hence,
$(tu)\sigma{\D} = t\sigma{\D}u\sigma{\D} \sqsupset^X
f(v_1',\dots,v_m') \Rnbe f(\seq[m]{v})\sigma{\D}$ by
nonversatility of $tu$ and \C{@\xFs}.
\item[\C{@\xV}] \smallskip
The same reasoning as in case \C{\xFb\xV} applies as
$uv$ is nonversatile.
\item[\C{\lambda{\subt}}] \smallskip
Let $\lambda x.t >^X v$ because $t[x/z] \geq_\tau^X v$ for a fresh $z$
with $\tau(x) = \tau(z)$. Given a substitution $\sigma$ away from $X$, we
may assume that $\sigma$ is away from $\SET{x}$ without loss of
generality. Furthermore, we can always choose a fresh $z'$ with
$\tau(z) = \tau(z')$ such that $\sigma$ is away from $\SET{z'}$ and
we still have $t[x/z'] \geq_\tau^X v$. The induction hypothesis
together with (\textasteriskcentered) yields
$t[x/z']\sigma{\D} \sqsupseteq_\tau^X v'$ for some
$v\sigma \Rsbe v' \Rnbe v\sigma{\D}$. Since $\sigma$ is away from
$\SET{x,z'}$ we have $t[x/z']\sigma = t\sigma[x/z']$ and therefore
$t\sigma{\D}[x/z'] \sqsupseteq_\tau^X v' \Rnbe v\sigma{\D}$.
Hence, $(\lambda x.t)\sigma{\D} = \lambda x.t\sigma{\D}
\sqsupset^X v' \Rnbe v\sigma{\D}$ by nonversatility of
$\lambda x.t$ and \C{\lambda{\subt}}.
\item[\C{\lambda{\subt}\eta}] \smallskip
Let $\lambda x.t >^X v$ because $t[x/z] \geq_\tau^X vz$ for a fresh $z$
with $\tau(x) = \tau(z)$. Given a substitution $\sigma$ away from $X$, we
may assume that $\sigma$ is away from $\SET{x}$ without
loss of generality. Furthermore, we can always choose a fresh
$z'$ with $\tau(z) = \tau(z')$ such that $\sigma$ is away from
$\SET{z'}$ and we still have $t[x/z'] \geq_\tau^X vz'$. The
induction hypothesis together with (\textasteriskcentered)
yields $t[x/z']\sigma{\D} \sqsupseteq_\tau^X v'w'$ for some
$v\sigma \Rsbe v' \Rnbe v\sigma{\D}$ and
$z'\sigma \Rsbe w' \Rnbe z'\sigma{\D}$. Since
$\sigma$ is away from $\SET{x,z'}$ we have $w' = z'$,
$t[x/z']\sigma = t\sigma[x/z']$ and therefore
$t\sigma{\D}[x/z'] \sqsupseteq_\tau^X v'z' \Rnbe
(vz')\sigma{\D}$. Monotonicity of $\sqsupseteq_\tau$ together with
$\alpha$-renaming and nonversatility of $\lambda x.t$ yields
$(\lambda x.t)\sigma{\D} = \lambda x.t\sigma{\D}
\sqsupseteq_\tau^X \lambda z'.v'z'$. If
$\lambda x.t\sigma{\D} \sqsupset_\tau^X \lambda z'.v'z'$ we are done
as $\lambda z'.v'z' \Rb{\eta} v' \Rnbe v\sigma{\D}$. Otherwise,
$\lambda x.t\sigma{\D} = \lambda z'.v'z' \sqsupset_\tau v'
\Rnbe v\sigma{\D}$ as $\sqsupset_\tau$ contains $\eta$.
\item[\C{\lambda{=}}] \smallskip
Let $\lambda x.t >^X \lambda y.v$ because $t[x/z] >^X v[y/z]$ for a fresh
$z$ with $\tau(x) = \tau(y) = \tau(z)$. Given a substitution $\sigma$ away
from $X$, we may assume that $\sigma$ is away from $\SET{x,y}$
without loss of generality. Furthermore, we can always choose a
fresh $z'$ with $\tau(z) = \tau(z')$ such that $\sigma$ is away
from $z'$ and we still have $t[x/z'] >^X v[y/z']$. The induction
hypothesis yields $t[x/z']\sigma{\D} \sqsupset^X v'[y/z']$ for
some $v[y/z']\sigma \Rsbe v'[y/z'] \Rnbe v[y/z']\sigma{\D}$.
Note that we
can immediately fix $v'[y/z']$ as the form of the term because
$y$ does not occur freely in $v[y/z']\sigma$ by assumption.
Since $\sigma$ is away from
$\SET{x,y,z'}$ we have $t[x/z']\sigma = t\sigma[x/z']$ as well as
$v[y/z']\sigma = v\sigma[y/z']$ and thus
$t\sigma{\D}[x/z'] \sqsupset^X v'[y/z'] \Rnbe v\sigma[y/z']$
as well as $v\sigma \Rsbe v' \Rnbe v\sigma{\D}$. Hence,
$(\lambda x.t)\sigma{\D} = \lambda x.t\sigma{\D} \sqsupset^X
\lambda y.v' \Rsbe \lambda y.v\sigma{\D} \Rnbe
(\lambda y.v)\sigma{\D}$ by nonversatility of $\lambda x.t$ and
\C{\lambda{=}}.
\item[\C{\lambda{\neq}}] \smallskip
Let $\lambda x.t >^X \lambda y.v$ because $\lambda x.t >^X v[y/z]$ for a
fresh variable $z$ with $\tau(x) \neq \tau(y) = \tau(z)$. Given a
substitution $\sigma$ away from $X$, we may assume that $\sigma$ is away
from $\SET{x,y}$ without loss of generality. Furthermore, we can always
choose a fresh $z'$ with $\tau(z) = \tau(z')$ such that $\sigma$ is
away from $\SET{z'}$ and we still have $\lambda x.t >^X v[y/z']$. The
induction hypothesis yields
$(\lambda x.t)\sigma{\D} \sqsupset^X v'[y/z']$ for some
$v[y/z']\sigma \Rsbe v'[y/z'] \Rnbe v[y/z']\sigma{\D}$.
Note that we
can immediately fix $v'[y/z']$ as the form of the term because
$y$ does not occur freely in $v[y/z']\sigma$ by assumption.
Since $\sigma$ is away from
$\SET{x,y,z'}$ we have $v[y/z']\sigma = v\sigma[y/z']$ and therefore
$v\sigma \Rsbe v' \Rnbe v\sigma{\D}$. Hence,
$(\lambda x.t)\sigma{\D} = \lambda x.t\sigma{\D} \sqsupset^X
\lambda y.v' \Rsbe \lambda y.v\sigma{\D} \Rnbe
(\lambda y.v)\sigma{\D}$ by nonversatility of $\lambda x.t$ and
\C{\lambda{\neq}}.
\item[\C{\lambda\xFs}] \smallskip
Let $\lambda x.t >^X f(\seq{v})$ because $f \in \xFs$
and $\lambda x.t >_\tau^X v_i$ for all $1 \leq i \leq m$. By the
induction hypothesis together with (\textasteriskcentered) we
obtain $(\lambda x.t)\sigma{\D} \sqsupset_\tau^X v_i'$ for
some $v_i\sigma \Rsbe v_i' \Rnbe v_i\sigma{\D}$ for all $1 \leq i \leq m$.
Hence, $(\lambda x.t)\sigma{\D} = \lambda x.t\sigma{\D} \sqsupset^X
f(v_1',\dots,v_m') \Rnbe f(\seq[m]{v})\sigma{\D}$ by
nonversatility of $\lambda x.t$ and \C{\lambda\xFs}.
\item[\C{\lambda\xV}] \smallskip
The same reasoning as in case \C{\xFb\xV} applies as $\lambda x.t$ is
nonversatile and we may assume that $\sigma$ is away from $\SET{x}$
without loss of generality.
\item[\C{\xFs{\subt}}] \smallskip
This is a special case of \C{\xFb{\subt}}.
\item[\C{\xFs{=}}] \smallskip
Here, we can reason as in case \C{\xFb{\subt}} but can do without
\lemref{StructSmImpliesSuccX} which is not available for small function
symbols: If $t_j\sigma{\D} \geq_\tau u_i'$ for some
$u_i\sigma \Rsbe u_i' \Rnbe u_i\sigma{\D}$ then also
$f(\bar{t}) \sqsupset^X u_i'$
by \C{\xFs{\subt}}. Furthermore, since we can also assume that
$f(\bar{t}) \sqsupset_\tau^X u_i''$ for some
$u_i\sigma \Rsbe u_i'' \Rnbe u_i\sigma{\D}$, we
obtain $f(\bar{t}) \sqsupset_\tau^X u_i'$ by (\textasteriskcentered).
\item[\C{\xFs{\succ}}]
The same reasoning as in case \C{\xFb{\succ}} applies by additionally
using (\textasteriskcentered).
\item[\C{\xFs@}] \smallskip
The same reasoning as in case \C{\xFb@} applies by additionally using
(\textasteriskcentered).
\item[\C{\xFs\xV}] \smallskip
The same reasoning as in case \C{\xFb\xV} applies. \qed
\end{itemize}
\end{proof}

The proof of the case for \C{\xFs{=}} shows why we have chosen a
weaker formulation of the rule \C{\xFs{=}}. This problem does not
occur in CPO because it can be directly proven that $s \sqsupset^X t$
implies $s\sigma \sqsupset^X t\sigma$. In our setting, we can only prove
that $s >^X t$ implies $s\sigma{\D} \sqsupset^X t'$ for some
$t\sigma \Rsbe t' \Rnbe t\sigma{\D}$ which means that we may get two
different witnesses from two applications of the induction
hypothesis as it is the case in \C{\xFb{=}} and \C{\xFs{=}}. For big
function symbols, this issue could be resolved by establishing
\lemref{StructSmImpliesSuccX} but for small function symbols one
would need a version of \C{\xFs{\subt}} which contains $\asubt$ as
it is the case for \C{\xFb{\subt}}.

\end{document}